%% file: main.tex
\def\confversion{0}
\def\ifconf{\ifnum\confversion=1}
\def\ifnotconf{\ifnum\confversion=0}

\documentclass[11 pt]{article}
\def\showauthornotes{1}
\def\showkeys{0}
\def\showdraftbox{0}
\def\confversion{0}
\def\widemargin{0}
\def\ipadcompile{0}

\input{macros}

\input{abbrev}

\input{defs}

\definecolor{cadmiumgreen}{rgb}{0.0, 0.42, 0.24}

\usepackage{graphicx}
\usepackage{booktabs}

\allowdisplaybreaks
\begin{document}
	
	\title{Low Soundness Linearity Testing on the Half-Slice}

\author{Haakon Larsen \thanks{The University of Iowa, \; \tt{haakon-larsen@uiowa.edu}.} \and
Tushant Mittal\thanks{Stanford University, \; {\tt{tushant@stanford.edu}}. TM is a postdoctoral fellow supported by the NSF grants CCF-2143246 and CCF-2133154. } \and Silas Richelson  \thanks{University of California, Riverside, \; \tt{silas@cs.ucr.edu}.}
\and Sourya Roy\thanks{The University of Iowa, \; \tt{sourya-roy@uiowa.edu}.}}

	\date{}
	
	\maketitle
	\draftbox
	\begin{abstract}	
	Let \(f: \slice \to \{ 0,1  \} \) be a Boolean function on the Boolean half-slice, \(\slice\), \ie elements of $\{0,1\}^n$ with Hamming weight $n/2$.  We show that if \(f(x)+f(y)=f(x+y)\) holds with probability $\frac{1+\delta}{2}$ over a uniform pair $(x,y)$ such that $x,y,x+y\in T$, then $f$ agrees with some linear function on at least $\frac{1+\delta}{2}-o(1)$ fraction of the points in $T$.  More generally, we show that if \(f\) passes the natural \(k\)-query BLR  test with probability  \( \frac{1+\delta}{2} \) for any $k\geq3$, then it must agree with some affine function at \(\frac{1+\delta^{\frac{1}{k-2}}}{2}-o(1)\)  fraction of the points in $T$.

\medskip	

The only other known linearity test for the slice in the low soundness regime (\ie when $\delta$ can be arbitrarily small) was given by Kalai, Lifshitz, Minzer, and Ziegler [FOCS'24].  Our result improves upon this result in two significant ways: firstly, it works for $k=3$ queries, instead of requiring $k\geq4$; secondly, our result is sharper, \eg when $k=4$, we are able to conclude an agreement of \(\frac{1+\sqrt{\delta}}{2}-o(1)\) instead of \(\frac{1+c\sqrt{\delta}}{2}\) for $c\approx.0035$.  In particular, our result matches (up to the $o(1)$ term) the conclusion one obtains over the full hypercube via the classical BLR analysis.

	\medskip	
	
	Our main technical contribution is a new dense model theorem using bounds on Krawtchouk polynomials.  Using these Krawtchouk polynomial bounds, we also obtain a simple $k$-query test ($k\geq 5$) that avoids any use of the dense model machinery. This simplified test naturally extends to the slice over the $q$-ary hypercube, giving the first such result over larger alphabets.
	\end{abstract}

\pagenumbering{gobble} 
\newpage
\pagenumbering{gobble} 
\tableofcontents
\newpage
\pagenumbering{arabic} 
\input{intro}

\input{kravchuk}

\input{approximator}

\input{general}

\input{qslice}

\newcommand{\etalchar}[1]{$^{#1}$}

\appendix
 \end{document}

%% file: macros.tex
\usepackage{xspace,enumerate}
\usepackage{amsmath,amssymb}
\usepackage{amsthm}
\ifnum\confversion=0
	\usepackage[toc,page]{appendix}
\fi
\usepackage[T1]{fontenc}
\usepackage[utf8]{inputenc}

\usepackage{thmtools}
\usepackage{thm-restate}
\usepackage{color,graphicx}
\usepackage{boxedminipage}
\usepackage{makecell}
\usepackage{tabularx,multirow}
\usepackage{csquotes}
\usepackage{mdframed}

\ifnum\widemargin=0
\usepackage[top=1in, bottom=1in, left=1in, right=1in]{geometry}
\else
\usepackage[top=1in, bottom=1in, left=1.25in, right=1.25in]{geometry}
\fi

\ifnum\showkeys=1
\usepackage[color]{showkeys}
\fi

\definecolor{darkred}{rgb}{0.5,0,0}
\definecolor{darkgreen}{rgb}{0,0.35,0}
\definecolor{darkblue}{rgb}{0,0,0.55}

\usepackage[dvipsnames]{xcolor}

\usepackage[pdfstartview=FitH,pdfpagemode=UseNone,colorlinks,linkcolor=NavyBlue,filecolor=blue,citecolor=OliveGreen,urlcolor=NavyBlue,pagebackref]{hyperref}

\usepackage{enumitem}
\usepackage{comment}
\usepackage[color=asparagus]{todonotes}

\usepackage{tikz-cd}

\usepackage[capitalise,nameinlink]{cleveref}

\usepackage{microtype}

\usepackage{mathtools,dsfont, bbm}

\usepackage{palatino}
\usepackage[scaled=.95]{helvet}

\ifnum\ipadcompile=1
\usepackage{eulervm}
\else
\usepackage{eulerpx}
\fi

\DeclareFontEncoding{LGR}{}{}
\DeclareSymbolFont{sfitgreek}{LGR}{cmss}{m}{it}
\SetSymbolFont{sfitgreek}{bold}{LGR}{cmss}{bx}{it}
\DeclareMathSymbol{\sfpi}{\mathord}{sfitgreek}{`p}

\usepackage{tcolorbox}

\DeclareMathAlphabet{\mathpazocal}{OMS}{zplm}{m}{n}
\DeclareRobustCommand*{\mathcal}[1]{\mathpazocal{#1}}
\ifnum\showauthornotes=1
\newcommand{\Authornote}[3]{{\sf\small\color{#3}{[#1: #2]}}}
\newcommand{\Authorcomment}[2]{{\sf \small\color{gray}{[#1: #2]}}}
\newcommand{\Authorfnote}[2]{\footnote{\color{red}{#1: #2}}}
\else
\newcommand{\Authornote}[3]{}
\newcommand{\Authorcomment}[2]{}
\newcommand{\Authorfnote}[2]{}
\fi

\ifnum\showdraftbox=1
\newcommand{\draftbox}{\begin{center}
  \fbox{%
    \begin{minipage}{2in}%
      \begin{center}%
        \begin{Large}%
          \textsc{Working Draft}%
        \end{Large}\\
        Please do not distribute%
      \end{center}%
    \end{minipage}%
  }%
\end{center}
\vspace{0.2cm}}
\else
\newcommand{\draftbox}{}
\fi

\declaretheorem[numberwithin=section]{theorem}
\declaretheorem[sibling=theorem]{lemma}
\declaretheorem[sibling=theorem]{claim}

\declaretheorem[sibling=theorem]{fact}
\declaretheorem[sibling=theorem]{corollary}

\theoremstyle{definition}
\declaretheorem[sibling=theorem]{definition}

\declaretheorem[sibling=theorem]{notation}

\newtheorem{algo}[theorem]{Algorithm}

\def\FullBox{\hbox{\vrule width 6pt height 6pt depth 0pt}}

\def\qed{\ifmmode\qquad\FullBox\else{\unskip\nobreak\hfil
\penalty50\hskip1em\null\nobreak\hfil\FullBox
\parfillskip=0pt\finalhyphendemerits=0\endgraf}\fi}

\def\qedsketch{\ifmmode\Box\else{\unskip\nobreak\hfil
\penalty50\hskip1em\null\nobreak\hfil$\Box$
\parfillskip=0pt\finalhyphendemerits=0\endgraf}\fi}

%% file: abbrev.tex
\DeclareMathOperator{\agr}{agr}

\newcommand{\agrHom}{\agr{(f,\phi)}}
\newcommand{\tagrHom}{\widetilde{\agr}{(f,\phi)}}

\newcommand{\iotaslice}[1]{\widehat{\iota}_{q,t}(#1)}

\newcommand{\slice}{T}

\newcommand{\qslice}{T_{\scriptscriptstyle \frac{q-1}{q}n}}

\newcommand{\qcube}{\mathbb{F}_q^n}

\newcommand{\krawt}{\mathcal{K}}

\newcommand{\indicator}[1]{\mathds{1}_{#1}\xspace}
\newcommand{\tindicator}[1]{\iota_{#1}}

\newcommand{\Hom}{\textup{Hom}}

\newcommand{\ep}{\varepsilon}

\makeatletter
\newcommand{\stackalign}[1]{
	\vcenter{
		\Let@ \restore@math@cr \default@tag
		\baselineskip\fontdimen10 \scriptfont\tw@
		\advance\baselineskip\fontdimen12 \scriptfont\tw@
		\lineskip\thr@@\fontdimen8 \scriptfont\thr@@
		\lineskiplimit\lineskip
		\ialign{\hfil$\m@th\scriptstyle##$&$\m@th\scriptstyle{}##$\crcr
			#1\crcr
		}
	}
}

\let\latexcirc=\circ
\newcommand{\ccirc}{\mathbin{\mathchoice
  {\xcirc\scriptstyle}
  {\xcirc\scriptstyle}
  {\xcirc\scriptscriptstyle}
  {\xcirc\scriptscriptstyle}
}}
\newcommand{\xcirc}[1]{\vcenter{\hbox{$#1\latexcirc$}}}\let\circ\ccirc

\def\to{\rightarrow}

\def\epsilon{\varepsilon}

\def\phi{\varphi}
\def\cal{\mathcal}

\def\implies{\Rightarrow}

\newcommand{\ie}{i.e.,\xspace}
\newcommand{\eg}{e.g.,\xspace}

\newcommand{\R}{{\mathbb R}}
\newcommand{\E}{{\mathbb E}}
\newcommand{\C}{{\mathbb C}}
\newcommand{\N}{{\mathbb{N}}}

\newcommand{\F}{{\mathbb F}}

\DeclarePairedDelimiter\parens{\lparen}{\rparen}
\DeclarePairedDelimiter\abs{\lvert}{\rvert}
\DeclarePairedDelimiter\norm{\lVert}{\rVert}

\DeclarePairedDelimiter\braces{\lbrace}{\rbrace}
\DeclarePairedDelimiter\brackets{\lbrack}{\rbrack}
\DeclarePairedDelimiter\angles{\langle}{\rangle}
\DeclarePairedDelimiterXPP\lnorm[1]{}\lVert\rVert{_2}{#1}

\DeclareMathDelimiter{\given}
      {\mathbin}{symbols}{"6A}{largesymbols}{"0C}

\newcommand{\prob}{\mathsf{Pr}}
\newcommand{\Esymb}{\mathbb{E}}

\newcommand{\Psymb}{\mathrm{Pr}}

\DeclarePairedDelimiterXPP{\Prob}[1]
 {\prob}{\lparen}{\rparen}{}
 {\renewcommand{\given}{\;\delimsize\vert\nonscript\;\mathopen{}}#1}

\makeatletter
\def\Pr#1{%
    \ProbabilityRender{\Psymb}{#1}%
}

\def\Ex#1{%
    \ProbabilityRender{\Esymb}{#1}%
}

\def\condPE#1#2{%
	\@ifnextchar\bgroup
	{\ConditionalProbabilityRender{\widetilde{\Esymb}}{#1}{#2}}
	{\ProbabilityRender{\widetilde{\Esymb}}{#1 \given #2}}
}

\def\ConditionalProbabilityRender#1#2#3#4{
	\renderwithdist{#1}{#2}{#3 \given #4}	
}

\def\ProbabilityRender#1#2{%
  \@ifnextchar\bgroup%
  {\renderwithdist{#1}{#2}}
   {\singlervrender{#1}{#2}}
}
\def\singlervrender#1#2{%
   {\mathchoice
       {{#1}\brackets*{#2}}
       {{#1}[ #2 ]}
       {{#1}[ #2 ]}
       {{#1}[ #2 ]}
   }
}
\def\renderwithdist#1#2#3{%
   \@ifnextchar\bgroup
   {\superfancyrender{#1}{#2}{#3}}
   {\mathchoice
      {\underset{#2}{#1}\brackets*{#3}}
      {{#1}_{#2}[ #3 ]}
      {{#1}_{#2}[ #3 ]}
      {{#1}_{#2}[ #3 ]}
     }
}
\def\superfancyrender#1#2#3#4#5{
   \ensuremath{\mathchoice
      {\underset{#1}{{#1}}\left#4 #3 \right#5}
      {{#1}_{#2}#4 #3 #5}
      {{#1}_{#2}#4 #3 #5}
      {{#1}_{#2}#4 #3 #5}
   }
}
\makeatother

 \newcommand\SetSymbol[1][]{%
     \nonscript\:#1\vert
     \allowbreak
     \nonscript\:
     \mathopen{}}
  \DeclarePairedDelimiterX\Set[1]\{\}{%
     \renewcommand\given{\SetSymbol[\delimsize]}
     #1
}

\newcommand{\ip}[2]{\angles*{#1 , #2}}

\newcommand{\calH}{{\cal H}}

\DeclareMathOperator{\Cay}{Cay}

\DeclareMathOperator{\Tr}{Tr}

\definecolor{asparagus}{rgb}{0.53, 0.66, 0.42}
\definecolor{cambridgeblue}{rgb}{0.64, 0.76, 0.68}
\definecolor{celadon}{rgb}{0.67, 0.88, 0.69}
\definecolor{charcoal}{rgb}{0.21, 0.27, 0.31}

%% file: defs.tex
\newif\ifnotes
\notestrue
\ifnotes

\newcommand{\bu}{{u}}
\newcommand{\bx}{{x}}
\newcommand{\by}{{y}}
\newcommand{\bz}{{z}}

\newcommand{\calK}{\krawt} %

\newcommand{\calO}{\mathcal{O}}
\newcommand{\calP}{\mathcal{P}}

\newcommand{\sfP}{{\sf P}}
\newcommand{\sfQ}{{\sf Q}}

\newcommand{\boldone}{{1\!\!1}}

\newcommand{\niceand}{\text{ }\&\text{ }}

\newcommand{\statement}{\underline{\sf statement}}
\newcommand{\zo}{\{0,1\}}

\newcommand{\fourierinner}[2]{\ip{\widehat{#1}^{#2}}{{\bf 1}}}

%% file: intro.tex
\section{Introduction}
\label{sec:intro}

Given oracle access to a Boolean function $f:\zo^n\rightarrow\zo$, we are interested in checking if $f$ is \emph{close} to being a linear function using a small number of queries.  In theoretical computer science, this question is popularly known as the \emph{linearity testing} problem.  It has been widely studied due to its key role in many diverse applications.  The original work in the area is the legendary Blum-Luby-Rubinfeld (BLR) test~\cite{BLR90} which states that if $f$ passes the check $f(\bx+\by)=f(\bx)+f(\by)$ with non-trivial probability over a random pair $(\bx,\by)$, then $f$ must correlate with some linear function.

Let $n\in\N$ be an even integer and let $|\bx|$ denote the Hamming weight of $\bx \in \zo^n$. The ``half slice'' of the Boolean hypercube $\zo^n$ is defined as,
\[
\slice ~:=~ \Big\{\,\bx\in\zo^n \,:\, |\bx| = \frac{n}{2} \,\Big\}.
\] 
  Motivated by questions in hardness amplification, PCP, and direct product testing, David, Dinur, Goldenberg, Kindler, and Shinkar~\cite{DDGKS17} initiated the study of linearity testing over $\slice$.  They proved that if $f:\slice\rightarrow\zo$ passes the BLR check $f(\bx+\by)=f(\bx)+f(\by)$ with probability $1-\ep$ over a random pair $(\bx,\by)$ such that $\bx,\by,\bx+\by\in \slice$, then $f$ agrees with some linear function on $1-\calO(\ep)$ fraction of the $\bx\in \slice$.  That is, they proved a linearity testing theorem over $\slice$ in the ``$99\%$ regime''.

The corresponding question in the ``$1\%$ regime'' asks whether a function which passes the above BLR test over $\slice$ with probability that is noticeably better than $1/2$ must have non-trivial correlation with a linear function.  This question was left open by~\cite{DDGKS17} and was essentially untouched until the recent work of Kalai, Lifshitz, Minzer, and Ziegler~\cite{KLMZ24} who proved that if $f:T\rightarrow\zo$ passes the $4$-query BLR check $f(\bx+\by+\bz)=f(\bx)+f(\by)+f(\bz)$ with probability $\frac{1+\delta}{2}$ over a random triple $(\bx,\by,\bz)$ such that $\bx,\by,\bz,\bx+\by+\bz\in T$, then $f$ must agree with some affine function on $\frac{1+c\sqrt{\delta}}{2}$ fraction of the inputs in $T$, for a constant $c>0$.

This was exciting progress, and the argument in~\cite{KLMZ24} is extremely elegant with several novel components.  However, two sub-optimal features of their result are 1) it uses four queries, rather than the original BLR test, which uses three; and 2) the constant $c$ is quite small ($c\approx.0035$).  In particular, even if $f$ passes the four query BLR test over $T$ with $99\%$ probability, the result of~\cite{KLMZ24} only guarantees that $f$ agrees with some linear function on $50.2\%$ of the inputs in $T$.  This is in contrast to the original BLR linearity test where it is known (via the analysis of~\cite{BCHKS96}) that if $f$ passes the $3$-query BLR test with probability $\frac{1+\delta}{2}$, then $f$ agrees with some linear function on a $\frac{1+\delta}{2}$ fraction of the inputs in $\zo^n$.  Our main result extends the BLR linearity test to the half slice.

\begin{restatable}[Three-Query Test]{theorem}{mainthree}\label{thm:main}
  Let $n\in\N$ be a multiple of $4$, and suppose $f:T\rightarrow\zo$ is such that \[
  \Pr{\bx,\by\sim\zo^n}{f(\bx+\by)=f(\bx)+f(\by)\;\big|\;\bx,\by,\bx+\by\in T} ~=~ \frac{1+\delta}{2}.
  \]
  Then, there exists some $\bu\in\zo^n$ such that $\Pr{\bx\sim\zo^n}{f(\bx)=\bx\cdot\bu\;\big|\;\bx\in T}~\geq~\frac{1+\delta}{2}-o(1)$.
\end{restatable}

\subsection{Outline} 
\newcounter{introeqs}
\setcounter{introeqs}{1}

\paragraph{Notation.} Let $h:\zo^n\rightarrow\R$ be any function. Given $\bu\in\zo^n$,  the Fourier coefficient, $\hat h(\bu)$, is defined as $\hat h(\bu)~:=~\E_{\bx\sim\zo^n}\bigl[h(\bx)\cdot(-1)^{\bx\cdot\bu}\bigr]$ . We will use the following shorthand: 
\[
	\ip{\hat h^k}{{\bf 1}}~:=~\sum_{\bu\in\zo^n}\hat h(\bu)^k. 
\]
 We write $A\lesssim B$ to indicate that $A\leq B+\eta$ for some quantity $\eta=o(1)$ which tends to $0$ as $n\rightarrow\infty$.

\paragraph{Plan.} The goal of this overview section is to introduce the two key ideas that we use (on top of the classical BLR analysis) to prove our main theorem.  They are:
\begin{enumerate}
\item some basic facts about the Fourier analysis of the half slice;
\item an ``approximator lemma'' which says that our function $f:T\rightarrow\zo$ is approximated by a global function $g:\zo^n\rightarrow\R$ defined on the entire hypercube.
\end{enumerate}

\noindent In the rest of this overview, we introduce these ideas and prove our main theorem assuming: 1) a specific bound on the Fourier coefficients of $T$ (equation (\ref{eq:introkrav}) below); and 2) our approximator lemma.  We conclude with a comparison to~\cite{KLMZ24} (which also used an approximator lemma).

\paragraph{The BLR Linearity Test.} To begin, we recall how the analysis for the $3$-query BLR test works on the full hypercube, as this forms the skeleton for the proof of Theorem~\ref{thm:main}.  Assume $f:\zo^n\rightarrow\zo$ is such that $\Pr{\bx,\by\sim\zo^n}{f(\bx+\by)=f(\bx)+f(\by)}=\frac{1+\delta}{2}$, and let $F:=(-1)^{f(\bx)}$.  Then, 
\[\delta ~=~ \E_{\bx,\by\sim\zo^n}\bigl[F(\bx)F(\by)F(\bx+\by)\bigr] ~=~ \ip{\hat F^3}{{\bf 1}}.
\] 
Since $\ip{\hat F^2}{{\bf 1}}=1$ by Parseval (as it is $\{\pm 1\}$-valued), we obtain:
\begin{equation}
 \max_{\bu} {\hat F}(\bu) ~\geq~ \frac{\ip{\hat F^3}{{\bf 1}}}{\ip{\hat F^{2}}{{\bf 1}}} ~=~  \delta \;.\end{equation}
But if $\bu\in\zo^n$ is such that $\hat F(\bu)\geq\delta$, then we have,
\[\Pr{\bx\sim\zo^n}{f(\bx)=\bx\cdot\bu} ~=~ \frac{1+\hat F(\bu)}{2} ~\geq~ \frac{1+\delta}{2}.\]  Therefore, if $f$ passes the $3$-query BLR test with probability $\frac{1+\delta}{2}$, then $f$ has at least as much agreement with some linear function.

\paragraph{A Simple Suboptimal $5$-query Linearity Test on the Half Slice} 
In order to illustrate the connection between linearity testing theorems on the slice and the Fourier analysis of $T$, we prove a simple but suboptimal $5$-query linearity testing theorem on the slice.  Let $k\geq5$ be a query parameter and $f:T\rightarrow\zo$ be a function that passes the $k$-query BLR test on $T$ with probability $\frac{1+\delta}{2}$. Define $F(\bx):=T(\bx)\cdot(-1)^{f(\bx)}$, where $T$ is the indicator function of the half slice normalized by its density, namely $T(\bx):=\tau^{-1}\boldone_{\bx\in T}$, where $\tau:= 2^{-n}\cdot\binom{n}{n/2}$. It is easy to see that,
\[
\Pr{x\sim \slice}{f(x)= {u\cdot x}} ~=~ \frac{1+\hat{F}(u)}{2}. 
\]
Therefore, finding a linear function with a large correlation is the same as finding the largest Fourier coefficient of $F$. A direct calculation yields the following useful relation that relates the test passing probability with the Fourier coefficients of $F$ and $T$: 
\begin{equation}\label{eq:ftdelta}
	\ip{\hat F^k}{{\bf 1}} ~=~ \delta\cdot\ip{\hat T^k}{{\bf 1}} .
\end{equation}
Therefore, for any $r$ such that $k \geq 2r +1$, one can obtain,
\begin{align}
\label{eq:strategy}		  \max_{\bu} {\hat F}(\bu)^{k-2r} ~\geq~ \delta\cdot \frac{\ip{\hat T^k}{{\bf 1}}}{\ip{\hat F^{2r}}{{\bf 1}}} \;.
\end{align}
If $k \geq 5$, we may pick $r=2$. Furthermore, using~\cref{eq:ftdelta}, we can approximate $\ip{\hat F^4}{{\bf 1}} \leq \ip{\hat T^4}{{\bf 1}}$.
\begin{align}
\label{eq:simple}	  \max_{\bu} {\hat F}(\bu)^{k-4} ~\geq~ \delta\cdot \frac{\ip{\hat T^k}{{\bf 1}}}{\ip{\hat F^{4}}{{\bf 1}}} ~\geq~ \delta\cdot \frac{\ip{\hat T^k}{{\bf 1}}}{\ip{\hat T^{4}}{{\bf 1}}}.
\end{align}
To handle the above expression, we observe that there is an exact expression for the Fourier coefficients of $T$ in terms of the \textit{Krawtchouk polynomials}, $\krawt_t(x)$. 

\paragraph{The Krawtchouk Polynomials and the Fourier Coefficients $\{\hat T(\bu)\}_{\bu}$.} Direct calculation shows that for $\bu\in\zo^n$, $\hat T(\bu)=\binom{n}{n/2}^{-1}\cdot\calK_{n/2}(|\bu|)$, where $|\bu|$ is the Hamming weight of $\bu$ and where $\calK_t(x)$ is the Krawtchouk polynomial: \[\calK_t(x):=\sum_{r=0}^{\min\{t,x\}}(-1)^r\binom{x}{r}\binom{n-x}{t-r}.\] The Krawtchouk polynomials form an orthogonal system with respect to the binomial distribution.  They are relevant in many contexts and are well studied.  Our focus on the middle slice means that in this work we will be mostly interested in the ``middle Krawtchouk polynomial'' $\calK_{n/2}(x)$.  Specifically, we prove in~\cref{sec:krav} (\cref{clm:krav}) that for all $k\geq3$, \begin{equation}
\sum_{x=1}^{n-1}\binom{n}{x}\cdot\bigg|\calK_{n/2}(x)\bigg/\binom{n}{n/2}\bigg|^k=o(1) \;\implies \; \ip{\hat T^k}{{\bf 1}} = 2 + o(1) \label{eq:introkrav}	
\end{equation}

As a direct corollary, we obtain a $5$-query linearity testing theorem by plugging the bound above into~\cref{eq:simple}. Thus, we have obtained a non-trivial test by using essentially nothing except some basic facts about the Fourier coefficients of $T$. This highlights the importance of understanding the Fourier analysis of the slice for trying to prove linearity testing theorems on the slice. Moreover, this proof readily generalizes to the slice over general finite fields, $\F_q^n$. Over general fields, the test is no longer the BLR test but a variant due to Kiwi~\cite{Kiwi03} (this is needed even for the full hypercube). Following a similar strategy as above, we show that Kiwi's test also generalizes to the slice over a general field (\cref{sec:general}, \cref{thm:sliceq}), but we will not discuss this in this outline.

Coming back to the Boolean setting, the drawback of this simple test is that it is not quantitatively optimal. It does not match the guarantee one gets from the $5$-query BLR analysis on the full hypercube where one concludes that some $\bu\in\zo^n$ exists with $\hat F(\bu)\geq\delta^{1/3}$. Nor does it improve on the result from~\cite{KLMZ24}, which uses one fewer query and establishes the stronger conclusion $|\hat F(\bu)|=\Omega(\sqrt{\delta})$.  Nevertheless, it demonstrates that a non-trivial linearity testing theorem on the slice can be obtained using nothing but some basic facts about the Fourier coefficients of $T$.  Indeed, the Fourier analysis of $T$ plays a critical role in our optimal linearity test on the slice. To overcome these drawbacks, we use a more involved approach.

\paragraph{Optimal Linearity Testing on the Half Slice via the Approximator Lemma} The above approach fails to work directly for $k \leq 4$ because then we are forced to pick $r=1$ in~\cref{eq:strategy}. However, for a sparse set $T$, such as the slice, 
\[\ip{\hat F^2}{{\bf 1}} ~\geq~ {\tau^{-1}} ~=~ \Omega(\sqrt{n})\] (compared to the situation on the full hypercube where we had $\ip{\hat F^2}{{\bf 1}}=1$).
This prevents the above strategy from working as is. 
To circumvent this, one can try to approximate $F$ by a bounded $g$. The overview of the approach is as follows. Let $g:\zo^n\rightarrow\R$ such that:   (1) $\ip{\hat F^k}{{\bf 1}}\lesssim\ip{\hat g^k}{{\bf 1}}$, (2) $\big|\hat F(\bu)-\hat g(\bu)\big|=o(1)$ for all $\bu\in\zo^n$, and (3) $\ip{\hat g^2}{{\bf 1}} \leq \calO(1)$. Then, one can carry out the above strategy using $g$ in the following manner:
\begin{align*}
\ip{\hat g^k}{{\bf 1}} ~\geq~  \ip{\hat F^k}{{\bf 1}} ~&=~ \delta\cdot\ip{\hat T^k}{{\bf 1}}&&[\text{Property\;} (1)]\\
	  \max_{\bu} {\hat{F}}(\bu)^{k-2} ~\approx~ \max_{\bu} {\hat g}(\bu)^{k-2} ~\geq~ \frac{\ip{\hat g^k}{{\bf 1}}}{\ip{\hat g^2}{{\bf 1}}}~&\geq~ \delta\cdot \frac{\ip{\hat T^k}{{\bf 1}}}{\ip{\hat g^2}{{\bf 1}}}&&[\text{Property\;} (2)]\\
	   ~&\geq~ \calO(\delta). &&[\text{Property\;} (3)\; \text{and \cref{eq:introkrav}}]
\end{align*}
The heart of our result is the following ``approximator lemma'' that constructs such a function $g$.

\begin{restatable}[Approximator Lemma]{lemma}{approximator}
 \label{lem:main}
 Let $n,k\in\N$ with $n$ even, $k\geq3$ be such that at least one of $k$ and $n/2$ is even. Let $F:\zo^n\rightarrow\R$ be any function such that $|F|=T$, where $T$ is the normalized indicator function of the half slice. Then there exists a function $g:\zo^n\rightarrow\R$ such that:
  \begin{enumerate}
  \item $\ip{\hat F^k}{{\bf 1}}\lesssim\ip{\hat g^k}{{\bf 1}}$;
  \item $\ip{\hat g^2}{{\bf 1}}\leq2$;
  \item $\big|\hat F(\bu)-\hat g(\bu)\big|=o(1)$ for all $\bu\in\zo^n$.
  \end{enumerate}	
\end{restatable}

As a direct consequence of this approximator, we obtain an optimal $k$-query test over the slice for any $k \geq 3$.

\paragraph{\cref{thm:main} (Restated for General $k$).}\emph{Let $n,k\in\N$ with $n$ even, $k\geq3$ be such that at least one of $k$ and $n/2$ is even. Let $f:T\rightarrow\zo$ be a function which passes the $k$-query BLR test on the slice with probability $\frac{1+\delta}{2}$, for $\delta>0$.  Then there exists a $ u \in  \{0,1\}^n$ and a $b \in \zo$, such that
 	\[ \Pr{\vec x \sim \slice}{f(x) ~=~ u \cdot x  + b} ~\geq~ \frac{1+{\delta}^{\frac{1}{k-2}}}{2} - o(1)\;.\]
}

\paragraph{Remarks.} Some remarks are in order.
\begin{enumerate}
\item The analysis for the $k$-query BLR test on the hypercube (\ie the generalization of the argument presented at the beginning of the section) promises agreement of $\frac{1+\delta^{\frac{1}{k-2}}}{2}$ with an \textit{affine function} rather than a linear function, \ie a function of the form $u\cdot x + b$ for some  $(\bu,b)\in\zo^n\times\zo$. Thus, our main theorem gives a linearity testing theorem on the half slice which matches, up to the $o(1)$ term, the corresponding theorem on the full hypercube, for any $k\geq3$.
  \item The requirement that at least one of $\frac{n}{2}$ and $k$ is even is necessary for the $k$-query BLR test to be defined on the half slice.  Indeed, the identity $|\bx+\by|=|\bx|+|\by|-2|\bx\wedge\by|$ implies that when $k$ is odd, $|\bx_1+\cdots+\bx_{k-1}|$ is even for all $\bx_1,\dots,\bx_{k-1}\in\zo^n$ of the same Hamming weight.  Thus, when $k$ is odd the only way that $\bx_1,\dots,\bx_{k-1},\bx_1+\cdots+\bx_{k-1}$ can all be in the half slice is if $\frac{n}{2}$ is even.
\end{enumerate}

\paragraph{Comparison with~\cite{KLMZ24}.} Kalai \emph{et al.}~\cite{KLMZ24} also used an approximator lemma based on a ``dense model theorem'' for the slice.  Specifically, they identified a subset $S\subseteq\zo^n$ of constant density which approximates $T$ in a certain technical sense, and then they invoked a general theorem of Dodos and Kanellopoulos~\cite{DK16} which converts such an $S$ into an approximator lemma.  Our discussion above about the Fourier coefficients of the half slice indicate that $\hat T(\bu)=1$ if $\bu={\bf 0},{\bf 1}$, and $\hat T(\bu)=o(1)$ when $\bu\notin\{{\bf 0},{\bf 1}\}$.  Thus, the function $S:\zo^n\rightarrow\R$ whose Fourier coefficients are \[\hat S(\bu)=\left\{\begin{array}{cl}1, & \bu={\bf 0},{\bf 1}\\ 0, & \bu\notin\{{\bf 0},{\bf 1}\}\end{array}\right.\] makes a natural choice for an approximation of $T$.  This function $S$ turns out to be the weighted indicator function of the set $S:=\big\{\bx\in\zo^n:|\bx|\equiv \frac{n}{2}\text{ mod }2\big\}$ of strings whose Hamming weight has the same parity of $n/2$, which has density $1/2$ in $\zo^n$.  It turns out, however, that obtaining an approximator lemma using the theorem from~\cite{DK16} requires starting with a function which approximates $T$ in a higher Gowers uniformity norm, which the above $S$ does not satisfy.  For this reason,~\cite{KLMZ24} uses instead the normalized indicator function of the set $\big\{\bx\in\zo^n:|\bx|\equiv \frac{n}{2}\text{ mod }2^{k-1}\big\}$ which has density roughly $2^{-(k-1)}$ in $\zo^n$.  This leads to extra complications (which~\cite{KLMZ24} handles with a beautiful argument) and is the source of both suboptimal aspects of their result.

In this work, we specialize the techniques from~\cite{DK16} directly to the half slice to establish our improved approximator lemma, which yields an optimal linearity testing theorem on the half slice.  We are able to show that the above set $S$ of strings whose weight has the same parity as $n/2$ approximates $T$ well enough to establish our approximator lemma.  This requires us to go beyond the argument in~\cite{DK16} and control terms which are bounded in the special case of the half slice, but which are not bounded in general.  Our argument makes heavy use of the Fourier analysis of $T$.  See~\cref{sec:krav} for a more detailed overview of this part of our result.

\subsection{Relevant Prior Work}
\paragraph{Linearity testing beyond Boolean hypercube}
Linearity testing has been studied in many domains apart from the Boolean hypercube. As mentioned,~\cite{DDGKS17,KLMZ24} studied this question over the Boolean slice. Researchers have also considered other groups, going beyond the traditional hypercube setting.  In this direction,  Kiwi~\cite{Kiwi03} gave a linearity tester for functions, \(f: \F^n_q \to \F_q\), over vector spaces. Interestingly, the classic BLR test is not known to be effective in this setting. Recently,~\cite{MR26} gave a general framework for studying the linearity testing problem over finite groups and gave various new testing results over cyclic groups,  finite simple groups, etc.

\paragraph{Boolean function analysis on the slice} In recent years, there has been an increased interest in extending various results from Boolean function analysis to the slice setting. For instance, Filmus and Ihringer~\cite{FI19} showed that constant-degree functions over the slice are Juntas, generalizing a result by Nisan and Szegedy~\cite{NS94}. The famous \textit{invariance principle} has also been ported to the slice setting by Filmus, Kindler, Mossel, and Wimmer~\cite{FKMW18}. To achieve this result, they used certain explicit orthogonal basis functions constructed by Filmus~\cite{Filmus16} in a previous work. In a separate research, \cite{Flimus16FKN} also proved a version of the prominent Friedgut--Kalai--Naor (FKN) theorem over the slice.

%% file: kravchuk.tex
\section{Bounds for Binary Krawtchouk Polynomials}
\label{sec:krav}
In this section, we will present a self-contained proof of~\cref{eq:introkrav}. We prove a bound over the general $q$-ary hypercube~(\cref{claim:qkraw_mean}) later, but to make the Boolean case self-contained, we give its proof here. This will suffice to give our simple $5$-query test.

To start off, we have that the following expression of the Fourier coefficients of the half-slice, $\ip{\hat{T}^k}{{\bf 1}}$, can be computed using the Krawtchouk polynomials as the slice is invariant under permutations (and therefore its indicator is a symmetric function). The following is the exact expression, which is a standard fact:
\begin{equation}\label{eqn:krawt_expression}
\hat T(\bu)~=~\binom{n}{n/2}^{-1}\cdot\krawt_{n/2}(|\bu|);\quad \text{where} \quad\krawt_t(x)~:=~\sum_{r=0}^{\min\{t,x\}}(-1)^r\binom{x}{r}\binom{n-x}{t-r}.
\end{equation}
 
The proof of \cref{eqn:krawt_expression} relies on the following straightforward bounds on these \textit{Krawtchouk polynomials}, $\krawt_{n/2}(x)$. 

\begin{claim}[Krawtchouk Bounds for Boolean Slice]
	\label{claim:boolkraw_mean} The following bounds hold:
	\begin{enumerate}
		\item For any odd weight $u$, $\hat{T}(u)= 0$.
		\item For $|u|=2, n-2$, the Fourier coefficient satisfies, $\;|\hat{T}(u) | ~\leq~ \calO\parens[\big]{\frac{1}{n}}$.
		\item For any $u$, the Fourier coefficients satisfy, $\; |\hat{T}(u) |^2 ~\leq~ \calO\parens{\sqrt{n}}  \cdot \binom{n}{|u|}^{-1}$ .
	\end{enumerate}	
\end{claim}
\begin{proof}
 Observe that, $\krawt_{n/2}(n-i) = \krawt_{n/2}(i)$ for any $i$ as the terms are the same. But from the expression, $\krawt_{n/2}(n-i) = (-1)^i\krawt_{n/2}(i)$. This implies that $\krawt_{n/2}(i) = 0$ for odd $i$. 
The second claim is also a direct calculation.
\begin{align*}
	\krawt_{n/2}(2) ~&=~    \binom{n-2}{n/2} (-1)^{n/2} + 2\binom{n-2}{n/2-1} (-1)^{n/2-1} + \binom{n-2}{n/2-2} (-1)^{n/2-2} \\[0.6em]
	~&=~    \binom{n-2}{n/2-1} \cdot(-1)^{n/2-1}\cdot\parens[\Bigg]{ \frac{n/2-1}{n/2} (-1) + 2 + \frac{n/2-1}{n/2} (-1) } \\[0.6em]
	~&=~   - \binom{n-2}{n/2-1} \cdot \frac{4}{n}\\[0.6em]
	\implies\;  |\hat{T}(u)| ~&=~ |\krawt_{n/2}(2)|\cdot {n \choose n/2}^{-1} ~\leq~ \calO \parens[\Big]{\frac{1}{n}} \; .
\end{align*}
The last claim follows from Plancharel~(\cref{def:fourier}) when applied to the function, $T$, which implies that,
\[ \ip{\hat{T}^2}{{\bf 1}} ~=~ \norm{T}_2^2 ~=~ \tau^{-1} ~=~ \calO\parens{\sqrt{n}}.\]
Since $\hat{T}(u)$ only depends on $|u|$, we can rewrite the Fourier expression by grouping together terms of the same weight and denoting by $\hat{T}(x)$, the coefficient for any vector $u$ such that $|u| =x$. Thus, 
\[\ip{\hat{T}^2}{{\bf 1}} ~=~\sum_{u \in \zo^n} \hat{T}(u)^2 ~=~  \sum_{x=0}^n {n \choose x} \cdot \hat{T}(x)^2  ~=~ \tau^{-1} ~=~ \calO\parens{\sqrt{n}}  \;. \]
	
	Since the terms are non-negative for each $x$, the above inequality holds for every summand.
\end{proof}

\begin{lemma}[Fourier Bound]\label{clm:krav}
	Let $n \geq 10^8$ and $k\geq 3$ be integers. Then we have,
    \[ \ip{\hat{T}^k}{{\bf 1}} ~=~ 2 + O(n^{-\frac{1}{2}})  .\] Or equivalently, for any $k\geq3$, $\ip{(\hat T-\hat S)^k}{{\bf 1}}~=~o(1)$.
\end{lemma}

\begin{proof}[Proof]  
	By~\cref{eqn:krawt_expression} we have that, 
\[
    \ip{\hat{T}^k}{{\bf 1}}  ~=~   \sum^{n}_{x=0} \binom{n}{x} \cdot \parens[\bigg]{\frac{\krawt_{n/2}(x)}{\binom{n}{n/2}}}^{k} ~=:~ \sum_{x=0}^n A_x .
\]
We will now use~\cref{claim:boolkraw_mean} to bound the terms. We assume here that $k\geq 3$.
\begin{align*}
    A_0 ~&=~ A_n ~=~  1,\\
     A_1 ~&=~ A_{n-1} ~=~ 0,\\
    A_2 ~&=~ A_{n-2} ~=~ O(n^{2-k}),\\
         A_3 ~&=~ A_{n-3} ~=~ 0,\\
    A_x ~&\leq~ n^{\frac{k}{2}} \cdot \binom{n}{x}^{1-\frac{k}{2}} \leq~ O\parens[\Big]{n^{\frac{k}{2}} \cdot n^{(1-\frac{k}{2})\min(x, n-x)}} , \qquad \forall\, x \in [4,n-4] ,\\
\implies \sum_{x=4}^{n-4}   A_x   ~&\leq~   O(n^{4-\frac{3k}{2}}) .
\end{align*}
Thus,  
\[\ip{\hat{T}^k}{{\bf 1}} ~=~ A_0+ A_n + \sum_{x=3}^{n-3}   A_x  ~=~ 2+ O(n^{4-\frac{3k}{2}}). \qedhere
\]
\end{proof}

%% file: approximator.tex
\section{Proving the Approximator Lemma}
\label{sec:proof}
In this section, we prove our main technical result $-$ the approximator lemma.

\approximator*

  \subsection{Overview of the Proof of Lemma~\ref{lem:main}}
  \label{sec:approxoverview}
  The proof of Lemma~\ref{lem:main} involves two main ingredients, encapsulated in Claims~\ref{clm:rttv} and~\ref{clm:dodos}, below.  Throughout this section, we let $S\subset\zo^n$ denote the set of strings whose Hamming weight has the same parity as $n/2$; we let $\calH$ be the set of functions $h:\zo^n\rightarrow\R$ which are $(k-1)$-fold convolutions of bounded functions, $h=h_2\ast\cdots\ast h_k$ where $h_2,\dots,h_k:\zo^n\rightarrow[-1,1]$; and we write $o(1)$ for quantities which tend to $0$ as $n\rightarrow\infty$ with $k$ fixed.

  The first main ingredient in the proof of~\cref{lem:main} establishes that $F$ can be approximated in a weaker sense than what is needed for~\cref{lem:main}.  This claim follows from the main theorem of~\cite{RTTV08}, we give the derivation in~\cref{sec:rttv}.
  
  \begin{claim}[{\bf Weak Approximator via Regularity}]
    \label{clm:rttv}
    Let $n,k\in\N$ be such that $n$ is even, $k\geq3$, and so that at least one of $n/2$ and $k$ is even.  Let $F:\zo^n\rightarrow\R$ be such that $|F|=T$.  Then there exists $g:\zo^n\rightarrow[-2,2]$ which is supported on $S$ (\emph{i.e.}, $g(\bx)=0$ when $\bx\notin S$) such that $\ip{F-g}{h}=o(1)$ for all $h\in\calH$.  Moreover, if $k=3$ then $g$ additionally satisfies $\ip{F-g}{F\ast F}=o(1)$.
  \end{claim}

  \noindent The function $g:\zo^n\rightarrow[-2,2]$ guaranteed by Claim~\ref{clm:rttv} already satisfies two of the properties needed for~\cref{lem:main}; property (1), that $\ip{\hat F^k}{{\bf 1}}\leq\ip{\hat g^k}{{\bf 1}}+o(1)$, is the only one which is still outstanding.  Additionally, Claim~\ref{clm:rttv} promises that for all $h=h_2\ast\cdots\ast h_k\in\calH$, \[\ip{(\hat F-\hat g)\hat h_2\cdots\hat h_k}{{\bf 1}}=\ip{F-g}{h}=o(1).\]  This should be thought of as guaranteeing that $g$ approximates $F$ in a weaker sense than what is needed for Lemma~\ref{lem:main}.  The second main ingredient promises that $g$ approximates $F$ in a stronger sense, which turns out to be enough for Lemma~\ref{lem:main}.

  \begin{claim}[{\bf Boosting the Approximator}]
    \label{clm:dodos}
  For all $t=1,\dots,k$, and $h_{t+1},\dots,h_k:\zo^n\rightarrow[-1,1]$, 
  \[\big|\ip{(\hat F-\hat g)^t\hat h_{t+1}\cdots\hat h_k}{{\bf 1}}\big| ~=~o(1).\]
  \end{claim}

  \noindent We prove Claim~\ref{clm:dodos} in Section~\ref{sec:dodos} using an argument similar to the one in~\cite{DK16}.  We now prove Lemma~\ref{lem:main} assuming Claims~\ref{clm:rttv} and~\ref{clm:dodos}.

  \begin{proof}[Proof of Lemma~\ref{lem:main}]
    Let $g:\zo^n\rightarrow[-2,2]$ be the function promised by Claim~\ref{clm:rttv}. 
Since $g$ is supported on $S$, we have $\ip{\hat g^2}{{\bf 1}}=\ip{g^2}{{\bf 1}}\leq2$, by Parseval and using that $|g|\leq2$ and $S$ has density $\frac{1}{2}$ in $\zo^n$.  Furthermore, for all $\bu\in\zo^n$ we have \[\big|\hat F(\bu)-\hat g(\bu)\big|=\big|\ip{F-g}{\chi_{\bu}}\big|=o(1),\] since $\chi_{\bu}=\chi_{\bu}\ast\cdots\ast\chi_{\bu}\in\calH$.
All that remains to prove now is the first property, \ie 
\[\ip{\hat F^k-\hat g^k}{{\bf 1}}~=~\sum_{r=0}^{k-1}\ip{(\hat F-\hat g)\hat F^r\hat g^{k-r-1}}{{\bf 1}}~=~o(1).\] 
 We prove by induction on $r$ that for all $r=0,\dots,k-1$, 
 \[
 \big|\ip{(\hat F-\hat g)^t\hat F^r\hat g^{k-r-t}}{{\bf 1}}\big|~=~o(1),
 \] holds for all $t=1,\dots,k-r$, which suffices.  The base case when $r=0$ follows readily from Lemma~\ref{clm:dodos}.  Indeed, as $(g/2):\zo^n\rightarrow[-1,1]$, we have that for all $t=1,\dots,k$, 
 \[
 \big|\ip{(\hat F-\hat g)^t\hat g^{k-t}}{{\bf 1}}\big|~=~2^{k-t}\cdot\big|\ip{(\hat F-\hat g)^t(\hat g/2)^{k-t}}{{\bf 1}}\big|=2^{k-t}\cdot o(1)~=~o(1).
 \]  Now, if we assume that $r\in\{0,1,\dots,k-2\}$ is such that $\big|\ip{(\hat F-\hat g)^t\hat F^r\hat g^{k-r-t}}{{\bf 1}}\big|=o(1)$ for all $t=1,\dots,k-r$, then for any $t\in\{1,\dots,k-r-1\}$, we have:
 \[\big|\ip{(\hat F-\hat g)^t\hat F^{r+1}\hat g^{k-r-t-1}}{{\bf 1}}\big|~=~\Big|\ip{(\hat F-\hat g)^t\hat F^r\hat g^{k-r-t}}{{\bf 1}}+\ip{(\hat F-\hat g)^{t+1}\hat F^r\hat g^{k-r-t-1}}{{\bf 1}}\Big|~=~o(1),\] and so our induction is established and Lemma~\ref{lem:main} follows.
  \end{proof}

  \subsection{Auxiliary Tools}
  \label{sec:auxtools}

  In the previous section, we reduced Lemma~\ref{lem:main} to Claims~\ref{clm:rttv} and~\ref{clm:dodos}, which we prove in Sections~\ref{sec:rttv} and~\ref{sec:dodos} below.  In this section, we prove some auxiliary results that we will need for both claims.

  \paragraph{Computations Involving $T\ast T$.} We will use the following two facts about the convolution $T\ast T$.
  \begin{claim}
    \label{clm:ttbound}
    Let $T$ be the normalized indicator function of the half-slice. Then, $T*T$ satisfies the following:
    \begin{enumerate}
    \item   if $\frac{n}{3}\leq|\bx|\leq\frac{2n}{3}$, then $(T\ast T)(\bx)<4$;
    \item   $\big\|(T\ast T)^2\big\|_2=\calO(1)$.
    \end{enumerate}
  \end{claim}

  \begin{proof}
    \textbf{Item (1)}\quad Fix $\bx\in\zo^n$ with $\frac{n}{3}\leq|\bx|\leq\frac{2n}{3}$, and assume $|\bx|$ is even, since otherwise $(T\ast T)(\bx)=0$. Denote the complement of $x$ by $\overline{\bx}$.  Note the equivalence:
\[
\by\in T\niceand\bx+\by\in T\Longleftrightarrow|\by\wedge\bx|=\tfrac{1}{2}|\bx|\niceand|\by\wedge\overline{\bx}|=\tfrac{1}{2}|\overline{\bx}|\;. \]
Recall that $\tau=2^{-n}\cdot\binom{n}{n/2}$ is the density of $T\subset\zo^n$.  It follows that
  \begin{align*}
    (T\ast T)(\bx)
    ~&=~
    \tau^{-2}\cdot\Ex{\by\sim\zo^n}{\boldone_{\by\in T}\cdot\boldone_{\bx+\by\in T}}\\[1.1pt]
    ~&=~\tau^{-2}\cdot2^{-n}\cdot\binom{|\bx|}{|\bx|/2}\cdot\binom{|\overline\bx|}{|\overline\bx|/2} \\[1.1pt]
    ~&\leq~
    \tau^{-2}\cdot\biggl(\frac{\pi|\bx|}{2}\biggr)^{-1/2}\cdot\biggl(\frac{\pi|\overline\bx|}{2}\biggr)^{-1/2} &&\parens[\Bigg]{2^{-k}\binom{k}{k/2}\leq\Bigl(\frac{\pi k}{2}\Bigr)^{-\frac{1}{2}}}\\[1.1pt]
    ~&\leq~\tau^{-2}\cdot\frac{6}{\pi n} &&\parens[\Big]{|\bx|,|\overline{\bx}|\geq\frac{n}{3}}\\[1.1pt]
    ~&\leq~2n\cdot\frac{6}{\pi n}=\frac{12}{\pi}<4 &&\parens[\Bigg]{\binom{n}{n/2}\geq2^{n-1}\cdot\Bigl(\frac{n}{2}\Bigr)^{-1/2},\text{ for }n\geq8}.
  \end{align*}

\noindent \textbf{Item (2)}   We have
  \begin{align*}
    \big\|(T\ast T)^2\big\|_2^2
    ~&=~
    \Ex{\bx\sim\zo^n}{(T\ast T)(\bx)^4}\\
    ~&=~
    \Ex{\bx\sim\zo^n}{\boldone_{\frac{n}{3}\leq|\bx|\leq\frac{2n}{3}}\cdot(T\ast T)(\bx)^4} +\Ex{\bx\sim\zo^n}{\boldone_{|\bx|<\frac{n}{3}\text{ OR }|\bx|>\frac{2n}{3}}\cdot(T\ast T)(\bx)^4}\\
    ~&\leq~
    4^4+\tau^{-4}\cdot\Pr{\bx\sim\zo^n}{|\bx|< \frac{n}{3} ~\text{ OR }~|\bx|> \frac{2n}{3}}~=~\calO(1).
  \end{align*}
  The inequality on the third line has used the claim from part (i), and that $T\ast T\leq\tau^{-1}$. The final equality has used that $\tau^{-4}=\Theta(n^2)$, while the probability is $2^{-\Omega(n)}$ by Chernoff's inequality.
\end{proof}

  \paragraph{Cauchy-Schwarz Bounds.} Claims~\ref{clm:rttv} and~\ref{clm:dodos} both require bounds based on Cauchy-Schwarz.  To avoid repetition, we collect here the ideas that are common to the proofs of both bounds.

  \begin{claim}
    \label{clm:firstcs}
    Let $r\in\N$, and let $\varphi,h_i,h'_i:\zo^n\rightarrow\R$ for $i=1,\dots,r$ be arbitrary functions.  Then \[\ip{\phi}{\prod_{i=1}^r(h_i*h'_i)}^2 ~\leq~ \angles[\bigg]{ \phi*\phi\;, \; \prod_{i=1}^r h_i*h_i} \cdot\prod_{i=1}^r\|h'_i\|_2^2\;.\]
  \end{claim}

  \begin{proof}
    We have
  \begin{align}
    \nonumber \ip{\phi}{\prod_i(h_i*h'_i)}^2 ~&=~ {\Ex{x}{\phi(x) \prod_i(h_i*h'_i)(x)}}^2={\Ex{x,z}{\phi(x+z) \prod_i(h_i*h'_i)(x+z)}}^2\\
    \nonumber	~&=~ {\Ex{x,z}{\phi(x+z) \cdot \prod_i \Ex{a_i}{h_i(x+a_i)h'_i(a_i+z)}}}^2\\
    \nonumber~&=~ {\Ex{ \{a_i\}, z}{ \Ex{x}{\phi(x+z) \cdot \prod_i h_i(x+a_i)} \cdot \prod_i h'_i(a_i+z)}}^2\\
    \nonumber~&\leq~ \Ex{ \{a_i\}, z}{ {\Ex{x}{\phi(x+z) \cdot \prod_i h_i(x+a_i)}}^2} \cdot \Ex{ \{a_i\}, z}{ \prod_i h'_i(a_i+z)^2}\\
    \nonumber~&=~\Ex{ \{a_i\}, z, x, x'}{ \phi(x+z)\phi(x'+z) \cdot \prod_i h_i(x+a_i)h_i(x'+a_i)} \cdot \prod_i \|h'_i\|_2^2 \\
        \nonumber~&=~ \angles[\bigg]{ \phi*\phi\;, \; \prod_{i=1}^r h_i*h_i} \cdot \prod_i \|h'_i\|_2^2 \;,
  \end{align}
  where the inequality is Cauchy-Schwarz.
  \end{proof}

  \begin{claim}
    \label{claim:conv_ind}
  Let $r\geq2$ and let $h_1,\dots, h_r:\zo^n\rightarrow\R$ be arbitrary functions.  Then 
  \[\big\lVert h_1\ast\cdots\ast h_{r}\big\rVert_2^4 ~\leq~ \prod_{i=1}^{r}	\big\lVert{h_i\ast h_i}	\big\rVert_2^2 ~\leq~ \prod_{i=1}^{r}	\big\lVert{(h_i\ast h_i)^2}	\big\rVert_2 \;.\]
\end{claim}

\begin{proof} We prove the first inequality via induction on $r$. The second follows by applying Jensen's inequality. The base case is when $r=2$.  For any $h,h':\zo^n\rightarrow\R$, Cauchy-Schwarz gives:
\begin{equation}\label{eq:inductive}
\big\lVert{h\ast h'}	\big\rVert_2^2 ~=~ \ip{h*h'}{h*h'} ~=~ \ip{h*h}{h'*h'} ~\leq~ \norm{h*h}_2\cdot\norm{h'*h'}_2.
\end{equation}
Consider now the case when $r>2$, and let $h=h_1$ and $h'=h_2\ast\cdots\ast h_r$.  We have \[\big\lVert{h_1\ast\cdots\ast h_r}	\big\rVert_2^2 ~=~ \big\lVert{h\ast h'}	\big\rVert_2^2 ~\leq~ \norm{h*h}_2\cdot\norm{h'*h'}_2\leq\norm{h*h}_2\cdot\norm{h'}^2_2 ~\leq~ \prod_{i=1}^{r}	\big\lVert{h_i\ast h_i}	\big\rVert_2,\] where the first inequality is~\cref{eq:inductive}, the second is Cauchy-Schwarz, and the last holds by induction.
\end{proof}

\subsection{Proof of Claim~\ref{clm:rttv} $-$ Weak Approximator via Regularity}\label{sec:rttv}
The key tool to construct the weak approximator is the following result of Reingold, Trevisan, Tulsiani, and Vadhan~\cite{RTTV08}:

\begin{theorem}[{\bf Main Theorem of \cite{RTTV08}}]
  \label{thm:rttv}
  Let $A:\zo^n\rightarrow\R_{\geq0}$ be such that $\E_{\bx}\bigl[A(\bx)\bigr]\leq1$, let $\calH\subset\big\{h:\zo^n\rightarrow[-1,1]\big\}$ be a family of bounded functions, and let $\nu:\zo^n\rightarrow\R_{\geq0}$ be such that $\E_{\bx}\bigl[\nu(\bx)\bigr]\leq1$, $A\leq\nu$ and so that for all $r\in\N$ and $h_1,\dots,h_r\in\calH$, \[\angles[\big]{(\nu-1) \;,\; h_1\cdots h_r}~=~o(1).\]  Then there exists $w:\zo^n\rightarrow[0,1]$ such that $\ip{(A-w)}{h}=o(1)$ for all $h\in\calH$.
\end{theorem}

\begin{corollary}
  \label{cor:rttv}
  Let $A$, $\calH$ and $\nu$ be as in the hypotheses of~\cref{thm:rttv}.  Let $\calP=\{\sfP_\alpha\}$ be a partition of $\zo^n$ into at most $\eta^{-1/2}$ sets, where $\eta>0$ is the $o(1)$ parameter in the conclusion of~\cref{thm:rttv}.  Let $\Psi:\zo^n\rightarrow\R$ be a function with the property that for each $\sfP_\alpha\in\calP$, there exists a sign $b_\alpha\in\{\pm1\}$ such that $\Psi(\bx)=b_\alpha\cdot A(\bx)$ for all $\bx\in\sfP_\alpha$.  Then there exists $w:\zo^n\rightarrow[-1,1]$ such that $\big|\ip{\Psi-w}{h}\big|=o(1)$ for all $h\in\calH$.  Furthermore, if $\varphi:\zo^n\rightarrow[-1,1]$ is a function which is constant on each $\sfP_\alpha\in\calP$, then $\big|\ip{\Psi-w}{\varphi}\big|=o(1)$ also holds.
\end{corollary}

\begin{proof}
  For each $\sfP_\alpha\in\calP$, define $\Psi_\alpha:\zo^n\rightarrow\R_{\geq0}$ as 
  \[
  \Psi_\alpha(\bx)~:=~ |\Psi(\bx)|\cdot\boldone_{\bx\in\sfP_\alpha} ~=~ A(\bx)\cdot\boldone_{\bx\in\sfP_\alpha} . \]
Thus, $\Psi_\alpha\leq A$ pointwise and so $\ip{\Psi_\alpha}{{\bf 1}}\leq1$ and $\Psi_\alpha\leq\nu$.  Thus, we can invoke~\cref{thm:rttv} on $\Psi_\alpha$, obtaining $w_\alpha:\zo^n\rightarrow[0,1]$ such that $\ip{\Psi_\alpha-w_\alpha}{h}=o(1)$ for all $h\in\calH$. 

   Now, let $w:\zo^n\rightarrow[-1,1]$ be the function defined by $w(\bx)=b_\alpha\cdot w_\alpha(\bx)$ for all $\bx\in\sfP_\alpha$.  For all $h\in\calH$, we have: 
   \[
   \big|\ip{\Psi-w}{h}\big|~=~\bigg|\sum_\alpha \ip{\boldone_{\sfP_\alpha}(\Psi-w)}{h}\bigg|~=~\bigg|\sum_\alpha b_\alpha\cdot\ip{\Psi_\alpha-w_\alpha}{h}\bigg|~\leq~\sum_\alpha\eta ~\leq~ \eta^{1/2}~=~o(1),
   \] where the second equality holds by definition of $\Psi_\alpha$ and $w_\alpha$, the first inequality holds by~\cref{thm:rttv}, and the second inequality holds because the partition size is at most $\eta^{-1/2}$.  Finally, if $\varphi:\zo^n\rightarrow[-1,1]$ is constant on each $\sfP_\alpha$, then we have \[\big|\ip{\Psi-w}{\varphi}\big|~=~\bigg|\sum_\alpha\ip{\boldone_{\sfP_\alpha}(\Psi-w)}{\boldone_{\sfP_\alpha}\varphi}\bigg|~=~\bigg|\sum_\alpha b_\alpha c_\alpha\cdot \ip{\Psi_\alpha-w_\alpha}{{\bf 1}}\bigg|\leq\eta^{1/2}~=~o(1),\] where $c_\alpha\in[-1,1]$ is the value taken by $\varphi$ on $\sfP_\alpha$.
\end{proof}

We can now prove Claim~\ref{clm:rttv}. Recall $T:\zo^n\rightarrow\R$ and $S:\zo^n\rightarrow\R$ are, respectively, the weighted indicator functions of the half-slice, and the set of strings whose Hamming weights have the same parity as $\frac{n}{2}$.

\paragraph{Claim~\ref{clm:rttv} (Restated).}\emph{Let $n,k\in\N$ be such that $n$ is even, $k\geq3$ and so that at least one of $n/2$ and $k$ is even.  Let $F:\zo^n\rightarrow\R$ be such that $|F|=T$.  Then there exists $g:\zo^n\rightarrow[-2,2]$ which is supported on $S$ (\emph{i.e.}, $g(\bx)=0$ when $\bx\notin S$) such that $\ip{F-g}{h}=o(1)$ for all $h\in\calH$.  Moreover, if $k=3$, then $g$ additionally satisfies $\ip{F-g}{F\ast F}=o(1)$.}

\begin{proof}
  Let $A=\frac{1}{2} T$, and $\nu=\frac{1}{2}\cdot(T-S)+1$.  Note that $A\leq\nu$ and $\ip{\nu}{{\bf 1}}=1$ both hold.  Let $\calH_k$ be the set of functions which are convolutions of $(k-1)$-bounded functions, \ie  \[
\calH_k = \braces[\big]{\, h_1*h_2*\cdots *h_{k-1} \mid h_i : \zo^n\to[-1,1] \,}.
\]  We will establish the following bound using Cauchy-Schwarz after completing the current proof.

\begin{fact}
  \label{fact:cs1}
For any $\varphi:\zo^n\rightarrow\R$, $r\in\N$ and $H_1,\dots,H_r\in\calH_k$, \[
\ip{\varphi}{H_1\cdots H_r}^4 ~\leq~ \ip{\hat\varphi^4}{{\bf 1}}.
\]
\end{fact}

Invoking Fact~\ref{fact:cs1} with $\varphi=\nu-1=\frac{1}{2}\cdot(T-S)$ gives \[\ip{\nu-1}{H_1\cdots H_r}~\leq~\frac{1}{16}\cdot\ip{(\hat T-\hat S)^4}{{\bf 1}}~=~o(1),\] using~\cref{clm:krav}.  Thus, $A,\calH_k,\nu$ satisfy the hypotheses of~\cref{thm:rttv}.  Now, let $\Psi=\frac{1}{2} F$, and let $\calP$ be the partition $\zo^n=\sfP^+\cup\sfP^-\cup\overline T$, where \[\sfP^+:=\big\{\bx\in T:F(\bx) >0 \big\};\text{ and }\sfP^-:=\big\{\bx\in T:F(\bx) < 0\big\}.\]  Note $\Psi(\bx)=A(\bx)$ holds when $\bx\in\sfP^+\cup\overline T$ and $\Psi(\bx)=-A(\bx)$ for $\bx\in\sfP^-$.  When $k\geq4$, this is the partition we will use. When $k=3$, we further subdivide this partition as follows.  Let $\eta>0$ be the $o(1)$ parameter from the conclusion of Theorem~\ref{thm:rttv}, and let $\zeta>0$ be a $o(1)$ parameter such that $24/\zeta+2\leq\eta^{-1/2}$.  Let $N=8/\zeta$ and let $\{I_i\}_{i=1,\dots,N}$ be a partition of the real interval $[-4,4]$ into $N$ sub-intervals, each of length $\zeta$.  So,
  \[
 [-4,4] ~=~ \bigsqcup_{i=1}^N I_i,\quad\text{where }   I_i ~=~ \bigl[-4+(i-1)\zeta,-4+i\zeta\bigr).
  \]  
 
 Now, let $\sfP^+,\sfP^-,\overline T$ be as above. For $i \in [N]$,  we define the partitions as follows:
 \begin{align*}
 \sfQ_i ~&:=~\big\{\bx\in\zo^n:(F\ast F)(\bx)\in I_i\big\},\\
 \overline T_i ~&:=~ \overline T\cap\sfQ_i,\\
 \sfP^{\pm}_i ~&:=~\sfP^{\pm}\cap\sfQ_i,\\
 	\sfP^{\pm}_{\sf rem} ~&:=~ \big\{\bx\in\sfP^{\pm}:|(F\ast F)(\bx)|>4\big\}.
 \end{align*}
 The overall partition of $\zo^n$ is then:
 \[\zo^n ~=~ \bigcup_{i=1}^N\sfP^+_i\cup\bigcup_{i=1}^N\sfP^-_i\cup\bigcup_{i=1}^N\overline T_i\cup\sfP^+_{\sf rem}\cup\sfP^-_{\sf rem}\; .\]  
 Note the size of the partition is $3N+2=24/\zeta+2\leq\eta^{-1/2}$, and moreover as this is a sub-partition of $\zo^n=\sfP^+\cup\sfP^-\cup\overline T$, for each partition set $\sfP$, either $\Psi\equiv A$ on $\sfP$ or $\Psi\equiv-A$ on $\sfP$. 
 
  Both partitions (this one as well as $\zo^n=\sfP^+\cup\sfP^-\cup\overline T$ for the $k\geq4$ case) satisfy the hypotheses of Corollary~\ref{cor:rttv}, and so in both cases we obtain a function $w:\zo^n\rightarrow[-1,1]$ such that $\big|\ip{\Psi-w}{h}\big|=o(1)$ for all $h\in\calH$.  Our final function $g$ is $S\cdot w$, from which it is clear that $g:\zo^n\rightarrow[-2,2]$ holds and that $g$ is supported on $S$.  Since $F$ is supported on $S$, we get that, $S\cdot\Psi=\frac{1}{2}S\cdot F=F$. Using this, we obtain that for all $h\in\calH_k$:
  \begin{align*}
  	\big|\ip{F-g}{h}\big| 
  	~&=~ \big|\ip{S\cdot\Psi-S\cdot w}{h}\big|&& (\text{\small{$S\cdot\Psi = F$}})\\
  	~&=~\big|\ip{\Psi-w}{S\cdot h}\big|\\
  	~&=~\Big|\ip{\Psi-w}{h}+\ip{\Psi-w}{\chi_{{\bf 1}}\cdot h}\Big| && (\text{\small{$S(\bx)=1+(-1)^{|\bx|}=1+\chi_{{\bf 1}}(\bx)$}})\\
  	~&=~o(1) \;. && (\text{\small{$\chi_{{\bf 1}}\cdot h\in\calH_k$}})
  \end{align*}
  To explain the final equality, observe that if $h=h_2\ast\cdots\ast h_k$, then $\chi_{{\bf 1}}\cdot h=(\chi_{{\bf 1}}\cdot h_2)\ast\cdots\ast(\chi_{{\bf 1}}\cdot h_k)$.
   
  Finally, it remains to show that when $k=3$, $\big|\ip{F-g}{F\ast F}\big|=o(1)$.  For this purpose, we let $\lfloor z\rfloor$ denote $z$ rounded down to the nearest multiple of $\zeta$ for $z\in[-4,4]$. We define $\varphi:\zo^n\rightarrow[-4,4]$ to be the function: \[\varphi(\bx)=\left\{\begin{array}{cl}\lfloor(F\ast F)(\bx)\rfloor & \text{ if }|(F\ast F)(\bx)|\leq 4, \\ 0 & \text{ if }|(F\ast F)(\bx)|> 4\;. \end{array}\right.\] So, in particular, for each $\bx\in\sfP^+_i\cup\sfP^-_i\cup\overline T_i$, $\varphi(\bx)=-4+(i-1)\zeta$.  Note $\varphi\equiv0$ on $\sfP^+_{\sf rem}\cup\sfP^-_{\sf rem}, $ and so $\varphi$ is constant on the partition sets, and thus $\big|\ip{\Psi-w}{\varphi}\big|=4\cdot\big|\ip{\Psi-w}{(\varphi/4)}\big|=o(1)$ by Corollary~\ref{cor:rttv}. 
  
  Note that $|\Psi-w|\leq|\Psi|+1=\frac{1}{2}T+1\leq\frac{1}{2}(T-S)+2$. Now, we have that $\big|\ip{\Psi-w}{F\ast F}\big|=o(1)$ by the following computation:
  \begin{align*}
    \big|\ip{\Psi-w}{F\ast F-\varphi}\big|
    ~&\leq~
    \ip{|\Psi-w|}{|F\ast F-\varphi|}\\[0.7em]
    ~&\leq~\big\langle\tfrac{1}{2}(T-S)+2,|F\ast F-\varphi|\big\rangle && \text{\small{$\parens[\big]{|\Psi-w|\leq\frac{1}{2}(T-S)+2}$}}\\[0.7em]
    ~&\leq~
    \big\langle\tfrac{1}{2}(T-S)+2,\zeta\big\rangle+\tfrac{1}{2\tau}\cdot\Ex{\bx\sim\zo^n}{\boldone_{|(F\ast F)(\bx)|>4}\cdot\big|(F\ast F)(\bx)\big|} && (\text{\small{Definition of $\varphi$}})\\[0.7em]
    ~&\leq~
    2\zeta+\tfrac{1}{2\tau}\cdot\Ex{\bx\sim\zo^n}{\boldone_{(T\ast T)(\bx)>4}\cdot(T\ast T)(\bx)} && (\text{\small{$|F\ast F|\leq T\ast T$}})\\[0.7em]
    ~&\leq~
    2\zeta+\tfrac{1}{2\tau^2}\cdot\Ex{\bx\sim\zo^n}{\boldone_{|\bx|<\frac{n}{3}\text{ OR }|\bx|>\frac{2n}{3}}} ~=~o(1).
  \end{align*}
  The first inequality on the last line holds because $T\ast T\leq\tau^{-1}$, and because $(T\ast T)(\bx)\leq4$ whenever $\frac{n}{3}\leq|\bx|\leq\frac{2n}{3}$, by~\cref{clm:ttbound}.  The last $o(1)$ equation holds because $\zeta=o(1)$ and because $\tau^{-2}=\Theta(n)$ while by Chernoff's inequality,
  \[\Pr{\bx\sim\zo^n}{|\bx|<\frac{n}{3}\,\text{ OR }\,|\bx|>\frac{2n}{3}} ~=~ 2^{-\Omega(n)}.\]  This along with the fact that $S\cdot(F\ast F) = 2(F\ast F)$,  allows us to conclude that $\big|\ip{F-g}{F\ast F}\big|=o(1)$ as \[\ip{F-g}{F\ast F}~=~\ip{S\cdot\Psi-S\cdot w}{F\ast F}~=~\ip{\Psi-w}{S\cdot(F\ast F)}~=~2\ip{\Psi-w}{F\ast F}\;. \qedhere\]
\end{proof}

\paragraph{\cref{fact:cs1} (Restated).}\emph{For any $\varphi:\zo^n\rightarrow\R$, $r\in\N$ and $H_1,\dots,H_r\in\calH_k$, \[\ip{\varphi}{H_1\cdots H_r}^4 ~\leq~ \ip{\hat\varphi^4}{{\bf 1}}.\]}

\begin{proof}
  Recall each $H_i\in\calH_k$ is a $(k-1)$-fold convolution of bounded functions, $H_i=h_{i,1}\ast\cdots\ast h_{i,k-1}$ for $h_{i,j}:\zo^n\rightarrow[-1,1]$.  Write $h_i=h_{i,1}$ and $h_i'=h_{i,2}\ast\cdots\ast h_{i,k-1}$ so that $H_i=h_i\ast h_i'$ for $i=1,\dots,r$.  We have
  \begin{align*}
    \ip{\varphi}{H_1\cdots H_r}^2 ~&=~ \ip{\varphi}{\prod_i(h_i\ast h'_i)}^2\\
    ~&\leq~ { \angles[\bigg]{ \phi*\phi\;, \; \prod_{i=1}^r h_i*h_i}} \cdot\prod_{i=1}^r\|h'_i\|_2^2\\
    ~&\leq~ \angles[\bigg]{ \phi*\phi\;, \; \prod_{i=1}^r h_i*h_i} \\
    ~&\leq~ \norm{\phi*\phi} \cdot \prod_{i=1}^r \norm{h_i*h_i}  \\
    ~&\leq~ \ip{\hat\phi^4}{{\bf 1}}^{\frac{1}{2}} \;,\end{align*} where the first inequality is~\cref{clm:firstcs}, the second holds because 
    \[
    \prod_{i=1}^r\big\|h'_i\big\|_2^4~=~\prod_{i=1}^r\big\|h_{i,2}\ast\cdots\ast h_{i,k-1}\big\|_2^4~\leq~\prod_{i=1}^r\prod_{j=2}^{k-1}\big\|h_{i,j}\ast h_{i,j}\big\|_2^2~\leq~\prod_{i=1}^r\prod_{j=2}^{k-1}\big\|h_{i,j}\big\|_2^4~\leq~1,
    \] (using \cref{claim:conv_ind}, Cauchy-Schwarz, and that each $\big\|h_{i,j}\big\|_2^4\leq1$), the third inequality is Cauchy-Schwarz and submultiplicativity of norms, the fourth holds because $\norm{h_i*h_i}\leq 1$, for each $i \in [r]$.
\end{proof}

\subsection{Proof of Claim~\ref{clm:dodos} $-$ Boosting the Approximator}
\label{sec:dodos}
In this section we prove \cref{clm:dodos}, completing the proof of \cref{lem:main} and \cref{thm:main}.  We follow the argument from~\cite{DK16}, taking advantage of the following extra properties of our approximator.

\begin{claim}[{\bf Extra Approximator Properties}]
  \label{clm:measure}
  Let $g$ be the function promised by~\cref{clm:rttv}.  We have:
  \begin{itemize}[align=left]
  \item[{\bf (P1)}] $\big\|\bigl((F-g)\ast(F-g)\bigr)^2\big\|_2~=~\calO(1)$;
  \item[{\bf (P2)}] ${\rm Pr}_{\bx\sim\zo^n}\Bigl[\big|(F-g)^{\ast(k-1)}(\bx)\big|>4^{k-1}+1\Bigr]~=~o(1)$.
  \end{itemize}    
\end{claim}

\begin{proof}
Recall that $|F-g|\leq T-S+4$ since $|F|=T$ and $g$ is supported on $S$, with $|g|\leq 2$.

\noindent	{\bf Establishing (P1).}   Note that $(T-S)\ast4=0$, and $T\ast S=S\ast S=S$, Thus,
  \[(F-g)\ast(F-g)~\leq~(T-S+4)\ast(T-S+4)=T\ast T-S+16~\leq~ T\ast T+16 \;.\] By using the bound $(A+B)^2~\leq~2(A^2+B^2)$ twice, we obtain: 
  \[
  \Bigl((F-g)\ast(F-g)\Bigr)^4~\leq~\bigl(T\ast T+16\bigr)^4~\leq~8\cdot\bigl((T\ast T)^4+16^4\bigr)\;.
  \] It follows from~\cref{clm:ttbound} that \[
  \big\|\bigl((F-g)\ast(F-g)\bigr)^2\big\|_2^2=\big\langle\bigl((F-g)\ast(F-g)\bigr)^4,{\bf 1}\big\rangle ~=~\calO\bigl(\ip{(T\ast T)^4}{{\bf 1}}\bigr)~=~\calO(1)\;.\]

\noindent  {\bf Establishing (P2).}  Note that \[\Ex{\bx\sim\zo^n}{(T-S+4)^{*(k-1)}(\bx)}~=~ \Ex{\bx\sim\zo^n}{(T-S+4)(\bx)}^{k-1}~=~4^{k-1}.\]  Thus,
  \begin{align*}
    \Pr{\bx\sim\zo^n}{\big|(F-g)^{\ast(k-1)}(\bx)\big|>4^{k-1}+1}
     ~&\leq~ \Pr{\bx\sim\zo^n}{(T-S+4)^{\ast(k-1)}(\bx)>4^{k-1}+1}\\[1.1pt]
    ~&\leq~
    \Pr{\bx\sim\zo^n}{\Bigl((T-S+4)^{\ast(k-1)}(\bx)-4^{k-1}\Bigr)^2>1}\\[1.2pt]
    ~&\leq~
    \Ex{\bx\sim\zo^n}{\bigl((T-S+4)^{\ast(k-1)}(\bx)\bigr)^2}-4^{2(k-1)} &&\text{(Markov)}\\
     ~&=~ \sum_{\bu\neq {\bf 0}}\bigl(\hat T(\bu)-\hat S(\bu)\bigr)^{2(k-1)}&&\text{(Parseval)}\\[1.1pt]
    ~&=~
    \ip{(\hat T-\hat S)^{2k-2}}{{\bf 1}} ~=~ o(1) &&\text{(\cref{clm:krav})}.
  \end{align*}
The Parseval is invoked for the function $H:=(T-S+4)^{\ast(k-1)}$, and uses that $\hat H(\bu)=\bigl(\hat T(\bu)-\hat S(\bu)\bigr)^{k-1}$ for $\bu\neq{\bf 0}$ and $\hat H({\bf 0})=4^{k-1}$.
\end{proof}

\paragraph{\cref{clm:dodos} (Restated).}\emph{For all $t=1,\dots,k$, and $h_{t+1},\dots,h_k:\zo^n\rightarrow[-1,1]$,} \[\big|\ip{(\hat F-\hat g)^t\hat h_{t+1}\cdots\hat h_k}{{\bf 1}}\big| ~=~o(1),\] where the $o(1)$ term tends to $0$ as $n\rightarrow\infty$ with $k$ fixed.

\begin{proof}
  For $t=1,\dots,k$, let $\statement(t)$ be shorthand for the statement: \[\sup_{h_{t+1},\dots,h_k}\Big\{\big|\ip{(\hat F-\hat g)^t\hat h_{t+1}\cdots\hat h_k}{{\bf 1}}\big|\Big\}=o(1),\] where the supremum is over all bounded functions $h_{t+1},\dots,h_k:\zo^n\rightarrow[-1,1]$.  
  
  Claim~\ref{clm:rttv} establishes that $\statement(1)$ holds. We will show $\statement(t)\Rightarrow\statement(t+1)$ for $t\in\{1,\dots,k-1\}$ which suffices.  Actually, our argument does not work for establishing $\statement(2)\Rightarrow\statement(3)$ when $k=3$.  This is due to our use of Cauchy-Schwarz in a particular place (Fact~\ref{fact:cs}) which turns out to be too loose for this specific case.  However, when $k=3$, we can establish $\statement(3)$ directly using the identity \[\ip{(\hat F-\hat g)^3}{{\bf 1}}=\ip{F-g}{F\ast F}-2\ip{(\hat F-\hat g)^2\hat g}{{\bf 1}}-\ip{(\hat F-\hat g)\hat g^2}{{\bf 1}}\] and all three terms on the right are $o(1)$ (the first is by Claim~\ref{clm:rttv}, and the second two are because $\statement(1)$ and $\statement(2)$ both hold).

  So now, assume that $\statement(t)$ holds for some $t\in\{1,\dots,k-1\}$, we establish $\statement(t+1)$.  And as addressed above, assume we are not in the specific case when $k=3$ and $t=2$.  Fix bounded functions $h_{t+2},\dots,h_k:\zo^n\rightarrow[-1,1]$, and let $G=(F-g)^{\ast t}\ast h_{t+2}\ast\cdots\ast h_k$, so that the quantity we must bound is $|\ip{F-g}{G}|$.  Note that $|F-g|\leq T-S+4$ (since $g$ is supported on $S$ and has $|g|\leq2$), and so Cauchy-Schwarz gives
  \begin{align*}
    \ip{F-g}{G}^2
    ~&\leq~
    \ip{T-S+4}{{\bf 1}}\cdot\ip{T-S+4}{G^2}\\
    ~&=~
    4\,\ip{T-S}{G^2}+16\,\ip{\boldone_{|G(\cdot)|\leq4^{k-1}+1}}{\,G^2} + 16\,\ip{\boldone_{|G(\cdot)|>4^{k-1}+1}}{\,G^2}\\
    ~&=:~ 4E_1+16E_2+16E_3.
  \end{align*}
  We prove the lemma by showing $E_1,E_2,E_3=o(1)$.  We will use the following bound which is derived from Cauchy-Schwarz.  We prove it below, after we finish the current proof.

  \begin{fact}
    \label{fact:cs}
    Let $\phi,h_2\dots,h_k:\zo^n\rightarrow\R$ and set $G=h_2\ast\cdots\ast h_k$.  Then \[\ip{\phi}{G^2}^4 ~\leq~\ip{\hat \phi^4}{{\bf 1}}\cdot \prod_{i=2}^{k-1}\big\|(h_i\ast h_i)^2\big\|_2^2 \cdot\Gamma_k,\] where $\Gamma_k:=\big\|(h_k\ast h_k)^2\big\|_2^2$ when $k\geq4$, and $\Gamma_3:=\big\|h_3\big\|_2^8$.
  \end{fact}

  \noindent We will invoke Fact~\ref{fact:cs} when each $h_2,\dots,h_k$ is such that either $|h_i|\leq1$ or else $h_i=F-g$.  In this special case, note that $\big\|(h_i\ast h_i)^2\big\|_2^2=\calO(1)$ for all $i=2,\dots,k$ (using property ${\bf (P1)}$ of~\cref{clm:measure} when $h_i=F-g$).  Moreover, since we are not considering the case of $\statement(3)$ when $k=3$ (recall we treated this case separately above), the $\Gamma_k$ term when $k=3$ is also $\calO(1)$ since $|h_3|\leq1$ in this case.  Thus, in all cases where we will use~\cref{fact:cs}, it will give \[\label{eq:cs}\ip{\phi}{G^2}^4=\ip{\hat\phi^4}{{\bf 1}}\cdot\calO(1).\tag{$*$}\]

\paragraph{Bounding $E_1$.} We have by (\ref{eq:cs}) and \cref{clm:krav}, \[E_1^4=\ip{T-S}{G^2}^4=\ip{(\hat T-\hat S)^4}{{\bf 1}}\cdot\calO(1)=o(1).\]

\paragraph{Bounding $E_2$.} Set $\varphi(\bx):=\boldone_{|G(\bx)|\leq4^{k-1}+1}\cdot\frac{G(\bx)}{4^{k-1}+1}$ so that $\varphi:\zo^n\rightarrow[-1,1]$.  We get:
\begin{align*}
  E_2
  &=\ip{\boldone_{|G(\cdot)|\leq4^{k-1}+1}}{G^2}=(4^{k-1}+1)\cdot\ip{\varphi}{G}=(4^{k-1}+1)\cdot\ip{(\hat F-\hat g)^t\hat\varphi\cdot\hat h_{t+2}\cdots\hat h_k}{{\bf 1}}\\
  &=
  (4^{k-1}+1)\cdot o(1)=o(1),
\end{align*}
since $\statement(t)$ holds.  We have used that $G=(F-g)^{\ast t}\ast h_{t+2}\ast\cdots\ast h_k$.

\paragraph{Bounding $E_3$.} If $t\in\{1,\dots,k-2\}$ then $h_k:\zo^n\rightarrow[-1,1]$ is bounded, and so 
\begin{align*}
  |G(\bx)|
  ~&\leq~
  \Ex{\by_2,\dots,\by_{k-1}\sim\zo^n}{\big|h_2(\by_2)\big|\cdots\big|h_{k-1}(\by_{k-1})\big|\cdot\big|h_k(\bx+\by_2+\cdots+\by_{k-1})\big|}\\
  ~&\leq~
  \prod_{i=2}^{k-1}\Ex{\by}{|h_i(\by)|} ~\leq~ 4^{k-2},
\end{align*}
since for $i=2,\dots,k-1$, either $|h_i|\leq1$, or $|h_i|\leq T-S+4$.  Thus, $E_3=\ip{\boldone_{|G(\cdot)|>4^{k-1}+1}}{G^2}=0$ in this case, as $\boldone_{|G(\cdot)|>4^{k-1}+1}\equiv0$.  So suppose $t=k-1$, in which case $G=(F-g)^{\ast(k-1)}$, and property ${\bf (P2)}$ of~\cref{clm:measure} gives 
\[
\Pr{\bx\sim\zo^n}{\big|G(\bx)\big|>4^{k-1}+1} ~=~ o(1)\;.
\]
  It thus follows from equation (\ref{eq:cs}) that 
  \[
  E_3^4 ~=~ \ip{\boldone_{|G(\cdot)|>4^{k-1}+1}}{G^2}=\ip{\hat\boldone_{|G(\cdot)|>4^{k-1}+1}^4}{{\bf 1}}\cdot\calO(1)=o(1)\cdot\calO(1)=o(1)\;. \qedhere
  \]
\end{proof}
  
\paragraph{\cref{fact:cs} (Restated).}\emph{Let $\phi,h_2\dots,h_k:\zo^n\rightarrow\R$ and set $G=h_2\ast\cdots\ast h_k$.  Then \[\ip{\phi}{G^2}^4 ~\leq~\ip{\hat \phi^4}{{\bf 1}}\cdot \prod_{i=2}^{k-1}\big\|(h_i\ast h_i)^2\big\|_2^2\cdot\Gamma_k,\] where $\Gamma_k:=\big\|(h_k\ast h_k)^2\big\|_2^2$ when $k\geq4$, and $\Gamma_3:=\big\|h_3\big\|_2^8$.}

\begin{proof}
  Write $h=h_2$ and $h'=h_3\ast\cdots\ast h_k$ so that $G=h\ast h'$.  We have
  \begin{align*}
    \ip{\phi}{G^2}^2
    ~&=~
    \ip{\phi}{(h\ast h')^2}^2\\
    ~&\leq~
    {\angles[\bigg]{ \phi*\phi\;, \; (h*h)^2}}\cdot\big\|h'\big\|_2^4 && (\text{\cref{clm:firstcs}})\\
    ~&\leq~  \ip{\hat\phi^4}{{\bf 1}}\cdot\big\|(h_2\ast h_2)^2\big\|_2^2\cdot\big\|h'\big\|_2^4 && \text{(Cauchy-Schwarz)}\;.
  \end{align*}
Finally, when $k\geq4$, we have \[\big\|h'\big\|_2^4 ~=~ \big\|h_3\ast\cdots\ast h_k\big\|_2^4 ~\leq~\prod_{i=3}^k\big\|(h_i\ast h_i)^2\big\|_2\] by~\cref{claim:conv_ind}, while when $k=3$, $\big\|h'\big\|_2^4=\big\|h_3\big\|_2^4$.  The fact follows.
\end{proof}

%% file: general.tex
\section{Linearity Testing on the $q$-slice}\label{sec:general}
Let $\F_q$ be a finite field of characteristic $p$, \ie $q = p^r$.  In this section, we focus on testing the proximity of a given function, $f: \F_q^n \to \F_q$, to $\mathrm{Lin}(\F_q^n,\F_q)$, the set of $\F_q$-linear maps from $\F_q^n$ to $\F_q$. These are maps of the form $\phi(x) = u\cdot x$ for some $u \in \F_q^n$.  The key result we will prove is the following:

\begin{restatable}[Test over $\F_q^n$]{theorem}{sliceq}
\label{thm:sliceq}Let $\qslice \subseteq \qcube$ be the set of vectors of Hamming weight $\frac{q-1}{q} n$.  	If for any odd $k \geq 5$, $f: \qslice\to \F_q$ passes \hyperref[test_gen]{$\F_q$-Test$_k$} with probability $\frac{1+ (q-1)\delta}{q}$, then there exists a $ u \in \F_q^n$, such that
\[ \Pr{\vec x \sim \qslice}{f(x) =  u\cdot x }~\geq~ \frac{1+ (q-1)\delta}{q} \,- o_n(1).  \] 
\end{restatable}

\subsection{Preliminaries}

\begin{definition}[Characters and Fourier Transform over $G$]\label{def:fourier} For every finite abelian group $G$, there exists a set (which is also a group) $\hat{G}$ of its \textit{characters}, \ie the set of linear functions, $\hat{G} = \{\chi: G \to \C \mid \chi \text{ is linear}\}$. Additionally, $G\cong \hat{G}$ as groups, and one has the general Fourier transform for any $f:G\to \C$,
\begin{align*}
\hat{f}(\chi) ~&:=~ \Ex{x\sim G}{f(x)\cdot \overline{\chi}(x)}\;,&&\text{[Fourier Coefficients]}\\
f(x) ~&=~ \sum_{\chi \in \hat{G}} \hat{f}(\chi) \chi(x)\;,&&\text{[Fourier Transform]}\\
\norm{f}_2^2 ~=~ \Ex{x \in G}{|f(x)|^2} ~&=~ \sum_{\chi \in \hat{G}} |\hat{f}(\chi)|^2 \;.&&\text{[Plancharel Theorem]}
\end{align*}

\end{definition}

Fix $\omega$, a primitive $p$-th root of unity. Define the trace map as,
\[
 \Tr: \F_q\to \F_p, \; \Tr(x) = \sum_{i=0}^{r-1} x^{p^i} \;.
\]
	
	\begin{fact}[Characters and Fourier Transform over $\F_q$]\label{fact:char_fq} The set of characters of the additive group $\F_q$, \ie linear functions from $\F_q \to \C^{\times}$, is given by $\{ \omega^{ \Tr(\beta \cdot)} \mid \beta \in \F_q\}$. Using these characters, the Fourier transform of the function $\indicator{z=0}:\F_q \to \C$ is given by,
	\[
	\indicator{z = 0} ~=~ \Ex{\beta\in \F_q}{\omega^{\Tr(\beta\cdot z)}}.
	\] 
	We will use $\psi: \F_q \to \C^{\times }$ to denote the composition map, $\psi(\cdot ):=\omega^{ \Tr(\cdot)}$. 
	\end{fact}

\begin{notation}[Shifted Agreement]
We will work throughout with the following quantity,
\[\tagrHom = \frac{q\agrHom - 1}{q-1}, \text{ where }, \agrHom ~=~ \Pr{x \in S}{ f(x) = \phi(x)}. 
\]
Note that this quantity depends on $S$, but we omit denoting it as it will be unambiguous. 
\end{notation}

In the Boolean case, $\tagrHom ~=~ \hat{F}(\phi)$. However, in the $\F_q$-case, it will be easier to work with $\tagrHom$ as the expression in terms of the Fourier coefficient is a bit cumbersome. However, we have the following expression:

\begin{corollary}\label{cor:agr_indicator}
	Let $S$ be any finite set and $f,\phi : S \to \F_q$ be any pair of functions. Then,
\[	\tagrHom ~=~\Ex{x\sim S}{\Ex{a \in \F_q^*}{\psi^{a}\parens[\big]{f(x)-\phi(x)}} }\;. \]
\end{corollary}
\begin{proof}
It is a direct computation using the above fact.
\begin{align*}
	\agrHom ~&=~ \Ex{x\sim S}{\indicator{f(x)-\phi(x) = 0}}\\
	~&=~\Ex{x\sim S}{ \Ex{\beta \in \F_q}{\psi^{\beta}\parens[\big]{f(x)-\phi(x)}} }&&\text{[\cref{fact:char_fq}]}\\
		~&=~\Ex{x\sim S}{ \frac{q-1}{q} \Ex{a \in \F_q^*}{\psi^{a}\parens[\big]{f(x)-\phi(x)}}  + \frac{1}{q} }.
\end{align*}
The last equality separates out the term corresponding to the trivial character, \ie $\beta = 0$. The proof follows by rearranging as $\tagrHom = \frac{q\agrHom-1}{q-1}$.  
\end{proof}		

Finally, we need an analog of the convolution trick in~\cref{eq:ftdelta}. 

\begin{fact}[Fourier Expression for cycles]\label{fact:pk}
Let $S\subseteq \F_q^n$ be any subset, and $S$, be the indicator of the set normalized by the density, \ie  $S = \frac {q^n}{\abs{S}} \cdot \indicator{S}$. Then,

\[
P_k(S) ~:=~ \underset{\vec x\sim S^k}{\mathrm{Pr}}\brackets[\big]{{\textstyle\sum_i x_i = 0}} ~=~ \fourierinner{S}{k}. \]
\end{fact}
\begin{proof}
	Let $T_S$ be the normalized adjacency matrix of $\Cay(\qcube, S)$. Then, its eigenvalues are $\{ \widehat{\tindicator{S}}(\chi) \mid \chi \in \widehat{\qcube} \}$. Now, since the graph is a Cayley graph, the number of $k$-cycles containing the identity is $\frac{\Tr(T_S^k)}{q^n}$  which is the same as $P_k(S)$. 
\end{proof}

\subsection{The test and its analysis}
Over general fields $\F_q$, the BLR test does not yield good agreement guarantees, even over the full hypercube $\F_q^n$. To remedy this, Kiwi~\cite{Kiwi03} defined a modification of the BLR test, and alternate Fourier-theoretic proofs for this test were given by~\cite{Gopalan13,GKS06}. Therefore, it is natural to use this test for functions over subsets of $\F_q^n$. 

\vspace{0.6em}
\begin{tcolorbox}[colframe=teal, colback=white, title={$\F_q$-$\mathrm{Test}_k$}, label={test_gen}]
	\begin{itemize}
		\item Sample $(a_1,\dots, a_k) \sim \F_q^*$ and $(x_1,\dots, x_k) \sim S^k$ such that $a_1x_1+\dots +a_kx_k=0$.
		\item If $a_1f(x_1)+\dots +a_kf(x_k)= 0$: return $1$; otherwise: return $0$.
	\end{itemize}
\end{tcolorbox}
\vspace{0.6em}

\paragraph{Test Distribution}

Let $\mu_k$ be the distribution over $(\F_q^*)^k\times S^k$ that is uniformly supported on elements, $((a_1,\dots, a_k), (x_1,\dots, x_k))$ such that $\sum_i a_ix_i = 0$. This is the distribution used to define the test.  We note two useful properties of this distribution.

\begin{claim}
 For any fixed $\vec{a} = (a_1,\cdots a_k) \in (\F_q^*)^k$, the distribution, $\mu_k$, satisfies:
	\begin{align}
\label{eqn:scalar_invariance}		\mu_k{(\vec{a}, \vec{x})} ~&=~  \mu_k{(\beta\vec{a}, \vec{x})}, \quad \text{for any }\beta \in \F_q^*, \;\vec{x} \in S^k.\\
\label{eqn:prob_invariance}\underset{\vec x\sim S^k}{\mathrm{Pr}}\brackets[\big]{\textstyle{\sum_i a_i x_i = 0}}	 ~&=~ \underset{\vec x\sim S^k}{\mathrm{Pr}}\brackets[\big]{{\textstyle\sum_i x_i = 0}} ~=~ \fourierinner{S}{k}. 
	\end{align}
\end{claim}
\begin{proof}
The first claim is direct as $\sum_i a_i x_i = 0$ if and only if $\sum_i \beta a_i x_i = 0$ as $\beta \neq 0$. For the second claim, we observe that there is a bijection matching the support of the two sides of the equation. Since the probability is uniform over the support, their total probability mass is the same. Let $\vec{x}$ be such that $\sum_i x_i = 0$, \ie $\mu(\vec{1},\vec{x}) \neq 0$. We map the vector to  $(a_1x_1,\dots, a_kx_k)$. This clearly satisfies that $\sum_i a_ix_i = 0$, and thus it lies in the support. The map is clearly injective, and it is surjective because $a_i \neq 0$, so it can be inverted. The last equality is again using the convolution expression as in~\cref{fact:pk}.
\end{proof}

\begin{claim}[Test Passing Probability]
			Let $S\subseteq \F_q^n$, $f:S \to \F_q$. Let  $\frac{1 + (q-1)\delta_k}{q}$ be the probability that $f$ passes the General Test. Then, we have
			\begin{align*}
				\delta_k ~=~\Ex{ (\vec{a}, \vec{x}) \sim \mu_k }{\;\psi \parens[\big]{\textstyle \sum_i a_if(x_i)} } \;. 
			\end{align*}
		\end{claim}
		\begin{proof} We start by rewriting the test probability using~\cref{fact:char_fq}.
\begin{align*}
\frac{1 + (q-1)\delta_k}{q} ~&=~\Ex{ (\vec{a}, \vec{x}) \sim \mu }{ \indicator{\sum_i a_if(x_i) = 0}}\\
 ~&=~ \Ex{ (\vec{a}, \vec{x}) \sim \mu }{ \Ex{\beta \sim \F_q}{\psi\parens[\big]{\textstyle \beta \sum_i  a_i f(x_i) }}}  &&\text{[\cref{fact:char_fq}]}\\
  ~&=~ \frac{q-1}{q} \Ex{ (\vec{a}, \vec{x}) \sim \mu }{ \Ex{\beta \sim \F_q^*}{\psi\parens[\big]{\textstyle \beta \sum_i  a_i f(x_i) }}} + \frac{1}{q}. &&\text{[Separating } \beta = 0]\\
	\delta_k		~&=~ \Ex{ (\vec{a}, \vec{x}) \sim \mu }{ \Ex{\beta \sim \F_q^*}{\psi\parens[\big]{\textstyle \beta \sum_i  a_i f(x_i) }}}
				&&\text{[Rearranging]}.
			\end{align*}
Now, we will use the invariance property of the distribution $\mu_k$ (denoted as $\mu$ for readability).
\begin{align*}
		\delta_k		~&=~ \Ex{ (\vec{a}, \vec{x}) \sim \mu_k }{ \Ex{\beta \sim \F_q^*}{\psi\parens[\big]{\textstyle \beta \sum_i  a_i f(x_i) }}}
				.\\
		~&=~ \Ex{\beta \sim \F_q^*}{ \Ex{ (\vec{a}, \vec{x}) \sim \mu }{\psi\parens[\big]{\textstyle \beta \sum_i  a_i f(x_i) }}}
				&&\text{[Fubini]}.\\
		~&=~ \Ex{\beta \sim \F_q^*}{ \Ex{ (\beta^{-1}\vec{a}, \vec{x}) \sim \mu }{\psi\parens[\big]{\textstyle \beta \sum_i  \beta^{-1}a_i f(x_i) }}}
				&&\text{[\cref{eqn:scalar_invariance}]}.\\
				~&=~  \Ex{ (\vec{a}, \vec{x}) \sim \mu }{\psi\parens[\big]{\textstyle \sum_i  a_i f(x_i) }}.&&\text{[Using~\cref{eqn:scalar_invariance} again]}.\qedhere
\end{align*}
		\end{proof}
We now write an exact expression for the $k$-moment of the agreement. This is the key lemma that enables the analysis of the modified BLR test. This is the $\F_q$-variant of~\cref{eq:ftdelta}.
 				
		\begin{lemma}[Agreement Expression]\label{claim:agr_expression} Let $f: S\to \F_q$, which we extend to $\qcube$ by defining it to be zero outside $S$.  
			\[
			\Ex{\phi \in  \mathrm{Lin}(\F^n_q,\F_q)}{\tagrHom^k} ~=~  \delta_k \cdot \fourierinner{S}{k}
			\]	
		\end{lemma}
		\begin{proof}
		We start by using the expression for the agreement from~\cref{cor:agr_indicator},
			\begin{align*}
				\tagrHom ~&=~ \Ex{x\sim S}{\Ex{a \in \F_q^*}{\psi\parens[\big]{af(x)-a\phi(x)} }} ,\\
				\tagrHom^k ~&=~ \Ex{\vec x\sim S}{\Ex{\vec a \in {\F_q^*}^k}{\psi\parens[\big]{\scriptstyle \sum_i a_i\parens{f(x_i)-\phi(x_i)}} }}.
			\end{align*}
Since $\psi$ is a character, it satisfies $\psi(a+b)= \psi(a)\cdot \psi(b)$. Using this, we get, 
\[		\Ex{\phi\in \mathrm{Lin}(\qcube,\F_q)}{\tagrHom^k} ~=~ \Ex{\stackalign{\vec x\sim S^k\\ \vec a \in {\F_q^*}^k} }{ \psi \parens[\big]{ {\scriptstyle \sum_i a_if(x_i)}}\cdot \Ex\phi{  \psi \circ \phi \parens[\big]{-{\scriptstyle \sum_i a_ix_i} }}}\;. \]
			Now, for any $v \in \F_q^n$, we have, 
			\begin{align*}
				\Ex{\phi\in \mathrm{Lin}(\qcube,\F_q)}{ \psi\circ \phi(v)} 
				~&=~ \begin{cases}
				 1 \quad \text{ if }\quad	v = 0, \\
				 \E_{\beta\in \F_q}\brackets{\psi(\beta)} ~=~ 0 \quad \text{ if }\quad	v \neq 0. \\
				\end{cases} .
			\end{align*}
The second equality above uses that the image of $v$ is uniform over $\F_q$ if $v\neq 0$. And then, from~\cref{fact:char_fq}, we get that this sum is $\indicator{1=0} = 0$. 			
			Plugging this back, we get,
			\begin{align*}
				\Ex{\phi\in \Hom(\F^n_q,\F_q)}{\tagrHom^k}  ~&=~
				\Ex{\stackalign{\vec x\sim S^k\\ \vec a \in {\F_q^*}^k } }{ \psi \parens[\big]{ {\scriptstyle \sum_i a_if(x_i)}} \cdot  \indicator{\sum_i a_i x_i = 0}}\\[0.7em]
				~&=~ \Pr{\vec x\sim S^k, \vec a \in {\F_q^*}^k}{{\scriptstyle\sum_i a_i x_i = 0}} \cdot 
				\Ex{ (\vec{a}, \vec{x}) \sim \mu }{ \psi \parens[\big]{\scriptstyle \sum_i a_if(x_i)}}
				\\[0.7em]
				~&=~ \fourierinner{S}{k} \cdot \delta_k \;. \quad \text{[\cref{eqn:prob_invariance}]} \qedhere
			\end{align*}
		\end{proof}

%% file: qslice.tex
\subsection{Cycle counts over the $q$-slice}
Let $\slice_{t}$ be the subset\footnote{We omit denoting $q$ here as we will work with a single fixed $q$ throughout the section.} of vectors of Hamming weight $t$ over $\F_q^n$. We will denote the Fourier coefficients of its normalized indicator by $\iotaslice{\chi}$ instead of $\widehat{\iota_{\slice_{t}}}(\chi)$ for brevity. It is a standard fact from the theory of association schemes that these Fourier coefficients are given by the \textit{$q$-Krawtchouk polynomials}.

\begin{fact}[Fourier Coefficients of $q$-slice]\textup{\cite[Prop 2.84]{BBIT21}} \label{fact:qslice_coeff} Let $n >0$ be an integer and $0 \leq k, t \leq n$. 
\[\widehat{\slice}_{t}(\chi) ~=~ \frac{\krawt_{q,t}(|\chi|)}{|\slice_{t}|}\quad, \text{where},\quad \krawt_{q,t}(x) = \sum_{i=0}^t  \binom{x}{i} \binom{n-x}{t-i} (1-q)^{t-i} \;.
\] 
Therefore, $\fourierinner{\slice_t}{k} =  \sum^{n}_{x=0} \binom{n}{x} (q-1)^x  \cdot \hat{\slice_t}(x)^k$.
\end{fact}

\paragraph{Some Bounds on Krawtchouk Polynomials}

\begin{fact}[Stirling's Bound]\label{fact:stirling}
For any $q,n \geq 2$, and $\alpha \in (0,1)$, we have, 
\[
\abs{\slice_{\alpha n}} ~=~\binom{n}{\alpha n} (q-1)^{\alpha n} ~\approx~ \frac{1}{\sqrt{2\pi \alpha (1-\alpha) n }} \cdot \parens[\Bigg]{\frac{(q-1)(1-\alpha)}{\alpha}}^{\alpha n} \cdot \frac{1}{(1-\alpha)^n}
\]
	In particular, for $\alpha = \frac{q-1}{q}$, one gets a bound of $\approx\frac{q^n}{\sqrt{n}}$, while for $\alpha = \frac{1}{2}$, one gets, $\approx\frac{\parens[\big]{2\sqrt{q-1}}^n}{\sqrt{n}}$.
\end{fact}

\begin{claim}[Bounds for the Mean-Slice]
	\label{claim:qkraw_mean} For $t = \frac{q-1}{q} n$, the following bounds hold:
	\begin{enumerate}
		\item For any $q$, $\hat{\slice_t}(1)= 0$, and thus for $q=2$, $\iotaslice{1} = \hat{\slice_t}(n-1) = 0$
		\item For $x=2, n-2$, the Fourier coefficient satisfies, $\;|\hat{\slice_t}(x) | ~\leq~ O\parens[\big]{\frac{1}{n}}$.
		\item For any $0\leq x \leq n$, the Fourier coefficients satisfy, $\;\hat{\slice_t}(x)^2 ~\leq~ O\parens{\sqrt{n}}  \cdot \binom{n}{x}^{-1} \cdot (q-1)^{-x}$ .
	\end{enumerate}	
\end{claim}
\begin{proof}
 Observe that for $q=2$, $\krawt_{2,t}(n-i) = \krawt_{2,t}(i)$ for any $i$ as the terms are the same. We now compute the level-1 coefficient.  
\begin{align*}
	\krawt_{q,t}(1) ~&=~    \binom{n-1}{t} (1-q)^{t} + \binom{n-1}{t-1} (1-q)^{t-1} , \\[0.7em]
	~&=~   \binom{n-1}{t} (1-q)^{t-1}\parens[\Big]{ (1-q) + \frac{t}{n-t} } ~=~ 0. && \brackets[\Big]{t ~=~ n- \frac{n}{q}}
\end{align*}

\begin{align*}
	\krawt_{q,t}(n-1) ~&=~    \binom{n-1}{t}  + \binom{n-1}{t-1} (1-q) \\[0.7em]
	~&=~   \binom{n-1}{t} q(2-q). \\[0.7em]
	 \abs{\iotaslice{n-1}} ~&=~ \frac{(2-q)}{(q-1)^{n-1}}. 
\end{align*}

The second claim is also a direct calculation.

\begin{align*}
	\krawt_{q,t}(2) ~&=~    \binom{n-2}{t} (1-q)^{t} + 2\binom{n-2}{t-1} (1-q)^{t-1} + \binom{n-2}{t-2} (1-q)^{t-2} \\[0.6em]
	~&=~    \binom{n-2}{t-1} \cdot(1-q)^{t-1}\cdot\parens[\Bigg]{ \frac{n-t-1}{t} (1-q) + 2 + \frac{t-1}{n-t} (1-q)^{-1} } \\[0.6em]
	~&=~    \binom{n-2}{t-1} \cdot(1-q)^{t-1}\cdot\parens[\Bigg]{ \frac{q-1}{t}  + \frac{1}{(q-1)n-t} }  && \brackets[\Big]{t ~=~ n- \frac{n}{q}}\\[0.6em]
	~&=~    -\binom{n-2}{t-1}\cdot (1-q)^{t-2}\cdot { \frac{q^2}{n} }, \\
	\implies\; \abs{\iotaslice{2}} ~&=~ \frac{1}{(q-1)(n-1)} \; .
\end{align*}
The above settles the case for $q=2$, as $\krawt_{2,t}(n-2) = \krawt_{2,t}(2)$.
We now compute the expression for a general $q$.
\begin{align*}
	\krawt_{q,t}(n-2) ~&=~    \binom{n-2}{t}  + 2\binom{n-2}{t-1} (1-q) + \binom{n-2}{t-2} (1-q)^{2} \\[0.6em]
	~&=~    \binom{n-2}{t-1} \cdot(1-q)\cdot\parens[\Bigg]{ \frac{n-t-1}{t} (1-q)^{-1} + 2 + \frac{t-1}{n-t} (1-q) } \\[0.6em]
		~&=~    \binom{n-2}{t-1} \cdot(1-q)\cdot\parens[\Bigg]{-(1-q)^{-2} +  \frac{q}{(q-1)^2n} + 2 - (q-1)^2 + \frac{q(q-1)}{n} } \\[0.6em]
\abs{\krawt_{q,t}(n-2)}	~&\leq~    2 \binom{n-2}{t-1} \cdot(q-1)^3\\[0.6em]
	\implies\; \abs{\iotaslice{n-2}} ~&\leq~ \frac{2n}{(q-1)^{t-2}(n-1)} ~\leq~ O\parens[\Big]{(q-1)^{-(t-2)}} \; .
\end{align*}
The last claim follows from Plancharel~(\cref{def:fourier}) when applied to the function, ${\slice_{q,t}} $.
	\[\sum_{x=0}^n {n \choose x} (q-1)^x \cdot \hat{\slice_t}(x)^2 ~=~ \norm{\slice_{q,\alpha n}}_2^2  ~=~ \frac{q^n}{|\slice_{q,\alpha n}|} ~\leq~ O\parens{\sqrt{n}}  \;. \]
The last inequality uses the Stirling bound~\cref{fact:stirling}. 	Since the terms are non-negative for each $x$, the above inequality holds for every summand. This implies our desired bound.
\end{proof}

We now prove our Fourier norm bound, using which we can finish up the proof as in the Boolean case.

\begin{lemma}\label{lem:cycles}
	Let $q\neq 2$ be a power of a prime, $n \geq 10^8$ and $k\geq 4$ be  integers. Then, \[ 
	\fourierinner{\slice_t}{k}  ~=~ 1+ O\parens[\big]{n^{3-k}}  .\]
\end{lemma}

\begin{proof}
	We will write $k=2m$ for some $m\geq 2$. 
	By \cref{fact:qslice_coeff} and \cref{fact:stirling} we have that, 
\[
   \fourierinner{\slice_t}{k} ~=~   \sum^{n}_{x=0} \binom{n}{x} (q-1)^x  \cdot \hat{\slice_t}{x}^{2m} ~=:~ \sum_{x=0}^n T_x .
\]
We will now use~\cref{claim:qkraw_mean} to bound the terms. We assume here that $m\geq 2$.
\begin{align*}
    T_0 ~&=~  1,\; T_1 ~=~ 0,\\
    T_n ~&=~ 2^{-\Omega(n)}\\
    T_2, T_{n-2} ~&=~ O(n^{2-2m})\\
    T_x ~&\leq~ n^m \cdot \parens[\Bigg]{\binom{n}{x} (q-1)^x}^{1-m} \leq~ O\parens[\Big]{n^{m} \cdot n^{(1-m)\min(x, n-x)}} , \qquad \forall\, x \in [3,n-3] ,\\
\implies \sum_{x=3}^{n-3}   T_x   ~&\leq~   O(n^{3-2m}) .
\end{align*}
Now, the result follows by summing up these terms.
\end{proof}

\begin{theorem}[Test for the $q$-Slice]
Let $\qslice \subseteq \F_q^n$ denote the mean-slice. If, for any odd $k \geq 5$,  a function $f: \qslice \to \F_q$ passes the \hyperref[test_gen]{$\F_q$-$\mathrm{Test}_k$} with probability $\frac{1+ (q-1)\delta}{q}$, then
	\[ \max_{\phi\in \widehat{\qcube}} \agrHom ~\geq~ \frac{1}{q}+ \frac{(q-1)}{q}\cdot \delta \cdot (1- o_n(1)).  \] 
	\end{theorem}
\begin{proof}
\cref{lem:cycles} implies that for any $r \geq 2$, we get
\[
 \frac{\fourierinner{\slice_t}{2r+1} }{\fourierinner{\slice_t}{2r}} ~\geq~ (1-o_n(1)).
\]
Let $k = 2r+1 \geq 5$. Now the main trick of the proof is to use the following inequality: 
	\begin{align*}
		 \max_{\phi\in \mathrm{Lin}(\qcube, \F_q)} \tagrHom ~&\geq~ \frac{
		  \sum_{\phi\in \widehat{\qcube}} \tagrHom^{2r+1}
		  }{ \sum_{\phi\in \widehat{\qcube}} \tagrHom^{2r}}\\[0.7em]
		   ~&=~ \frac{\delta_{2r+1}\fourierinner{\slice_t}{2r+1} }{\delta_{2r}\cdot\fourierinner{\slice_t}{2r}} && [\text{\cref{claim:agr_expression}}]\\[0.7em]
		   ~&\geq~  \delta_{2r+1} \cdot \frac{\fourierinner{\slice_t}{2r+1} }{\fourierinner{\slice_t}{2r}}    && [\delta_{2r} \leq 1] \\[0.7em]
		   ~&\geq~ (1-o_n(1)) \cdot \delta_{2r+1}.
		 \end{align*}		 
Thus, we get $\frac{q\max_\phi \agrHom -1}{q-1} ~\geq~  (1-o_n(1)) \cdot  {\delta_{k}}$, which yields the claim.  
\end{proof}